\documentclass[11pt,letterpaper]{article}

\usepackage{fullpage,amsmath,amssymb,amsthm}
\usepackage{xspace}
\usepackage{multirow}
\usepackage[hypertex]{hyperref}

\usepackage{amsmath}

\usepackage{comment}
\usepackage{color}
\usepackage{paralist}

\usepackage{graphicx}
\usepackage{ifpdf}

\newtheorem{theorem}{Theorem}[section]
\newtheorem{lemma}[theorem]{Lemma}

\theoremstyle{definition}

{\bfseries}{\rmfamily}
{\bfseries}{\rmfamily}

\newcommand\remove[1]{}

\newcounter{note}[section]

\def\polylog{\operatorname{polylog}}

\usepackage{url}

\makeatletter
\def\url@leostyle{%
  \@ifundefined{selectfont}{\def\UrlFont{\sf}}{\def\UrlFont{\small\ttfamily}}}
\makeatother
\begin{document}

\title{iBGP and Constrained Connectivity}
\author{Michael Dinitz\\Weizmann Institute of Science \and Gordon Wilfong\\Bell Labs}

\begin{titlepage}
\maketitle
\thispagestyle{empty}

\begin{abstract}
We initiate the theoretical study of the problem of minimizing the size of an iBGP overlay in an Autonomous System (AS) in the Internet subject to a natural notion of correctness derived from the standard ``hot-potato'' routing rules.  For both natural versions of the problem (where we measure the size of an overlay by either the number of edges or the maximum degree) we prove that it is NP-hard to approximate to a factor better than $\Omega(\log n)$ and provide approximation algorithms with ratio $\tilde{O}(\sqrt{n})$.  This algorithm is based on a natural LP relaxation and randomized rounding technique inspired by the recent work on approximating directed spanners by Bhattacharyya et al.~[SODA 2009], Dinitz and Krauthgamer [STOC 2011], and Berman et al.~[ICALP 2011].  In addition to this theoretical algorithm, we give a slightly worse $\tilde{O}(n^{2/3})$-approximation based on primal-dual techniques that has the virtue of being both fast (in theory and in practice) and good in practice, which we show via simulations on the actual topologies of five large Autonomous Systems.

The main technique we use is a reduction to a new connectivity-based network design problem that we call \emph{Constrained Connectivity}.  In this problem we are given a graph $G=(V,E)$, and for every pair of vertices $u,v \in V$ we are given a set $S(u,v) \subseteq V$ called the \emph{safe set} of the pair.  The goal is to find the smallest subgraph $H=(V,F)$ of $G$ in which every pair of vertices $u,v$ is connected by a path contained in $S(u,v)$.  We show that the iBGP problem can be reduced to the special case of Constrained Connectivity where $G = K_n$ and safe sets are defined geometrically based on the IGP distances in the AS.  Indeed, our algorithmic upper bounds generalize to Constrained Connectivity on $K_n$, and our $\Omega(\log n)$-lower bound for the special case of iBGP implies hardness for the general case.  Furthermore, we believe that Constrained Connectivity is an interesting problem in its own right, so provide stronger hardness results ($2^{\log^{1-\epsilon} n}$-hardness of approximation based on reductions from Label Cover) and integrality gaps ($n^{1/3 - \epsilon}$ based on random instances of Unique Games) for the general case.  On the positive side, we show that Constrained Connectivity turns out to be much simpler for some interesting special cases other than iBGP: when safe sets are symmetric and hierarchical, we give a polynomial time algorithm that computes an optimal solution.
\end{abstract}

\end{titlepage}

\section{Introduction}

The Internet consists of a number of interconnected subnetworks
called Autonomous Systems (ASes).
As described in~\cite{basu}, the way that routes to a given
destination are chosen by routers within an AS can be viewed as follows.
Routers have a ranking of routes based on economic considerations of the AS.
Without loss of generality, in what follows
we assume that all routes are equally ranked.
Thus routers must use some tie-breaking scheme in order to choose a route
from amongst the equally ranked routes.
Tie-breaking is based on traffic engineering considerations and
in particular, the goal is to get packets out of the AS as quickly
as possible (called {\em hot-potato routing}).

An AS attempts to achieve hot-potato routing using iBGP, the version of
the interdomain routing protocol BGP~\cite{stewart:99} used by routers within
a subnetwork to announce routes to each other that have been
learned from outside the subnetwork.  An iBGP configuration is defined by a \emph{signaling graph}, which is supposed to enforce hot-potato routing.  Unfortunately, while iBGP has many nice properties that make it useful in practice, constructing a good signaling graph turns out to be a computationally difficult problem.  For example, it is not clear \emph{a priori} that it is even possible to check in polynomial time that a signaling graph is correct, i.e.~it is not obvious that the problem is even in NP!  In this paper we study the problem of constructing small and correct signaling graphs, as well as a natural extension to a more general problem that we call \emph{Constrained Connectivity}.

\subsection{iBGP}
At a high level, iBGP works as follows.
The routers that initially know of a route are called {\em border routers}.
(These initial routes are those learned by the border routers from
routers outside the AS.)
The border router that initially knows of a route
is said to be the {\em egress router} of that route.
Each border router knows of at most one route.
Thus an initial set of routes $F$ defines a set of egress routers $X_F$ where
there is a one-to-one relationship between routes in $F$ and routers
in $X_F$.
The AS has an underlying physical network with edge weights (e.g., IGP
distances or OSPF weights).
The {\em distance} between two routers is then defined to be the
length of the shortest path (according to the edge weights) between them.
Given a set of routes, a router will rank highest the one whose
egress router is closest according to this definition of distance.
The {\em signaling graph} $H$ is an overlay network whose nodes represent
routers and whose edges represent the fact that
the two routers at its endpoints use iBGP to inform one another of
their current chosen route.
The endpoints of an edge in $H$ are called {\em iBGP neighbors}.
A path in $H$ is called a {\em signaling path}.  Note that iBGP neighbors are not necessarily neighbors in the underlying graph, since $H$ is an overlay and can include any possible edge.

Finally, iBGP can be thought of as working as follows:
in an asynchronous fashion, each router considers all the latest routes
it has heard about from its iBGP neighbors,
chooses the one with the closest egress router and tells its
iBGP neighbors about the route it has chosen.
This continues until no router learns of a route whose egress router is
closer than that of its currently chosen route.
When this process ends the route chosen by router $r$ is denoted by $R(r)$.
Let $P(r)$ be
the shortest path from $r$ to $E(r)$, the egress router of $R(r)$.
When a packet arrives at $r$, it sends it to the next router $r'$ on $P(r)$,
$r'$ in turn sends the packet to the next router on $P(r')$ and so on.
Thus if $P(r')$ is not the subpath of $P(r)$ starting at $r'$ then the
packet will not get routed as $r$ expected.

A signaling graph $H$ has the \emph{complete visibility property} for a set of egress routers $X_F$ if each router $r$ hears about (and hence
chooses as $R(r)$) the route in $F$ whose egress router $E(r)$ is closest to
$r$ from amongst all routers in $X_F$.  It is easy to see that $H$ will achieve hot-potato routing for $X_F$ if and only if it has the complete visibility property for $X_F$.  So we say that a signaling graph is \emph{correct} if it has the complete visibility property for all possible $X_F$.

Clearly if $H$ is the complete graph then $H$ is correct.  Because of this, the default configuration of iBGP and the original standard was to maintain a complete graph, also called a full mesh~\cite{stewart:99}.
However the complete graph is not practical and so network managers have
adopted various configuration techniques to reduce the size of the signaling
graph~\cite{rfc:2796,rfc:3065}.
Unfortunately these methods do not guarantee correct signaling
graphs~\cite{basu,griffinwilfong:2002a}.
Thus our goal is to determine correct signaling graphs with fewer
edges than the complete graph.  Slightly more formally, two natural questions are to minimize the number of edges in the signaling
graph or to minimize the maximum number of iBGP neighbors for any router
while guaranteeing correctness.
We define {\sc iBGP-Sum} to be the problem of finding
a correct signaling graph with the fewest edges, and similarly we define {\sc iBGP-Degree} to be the problem of finding
a correct signaling graph with the minimum possible maximum degree.

\subsection{Constrained Connectivity}
All we know \emph{a priori} about the complexity of {\sc iBGP-Sum} and {\sc iBGP-Degree} is that they are in $\Sigma_2$ (the second existential level of the polynomial hierarchy), since the statement of correctness is that ``there exists a small graph $H$ such that for all possible subsets  $X_F$ each router hears about the route with the closest egress router".  In particular, it is not obvious that these problems are in NP, i.e.~that there is a short certificate that a signaling graph is correct.  However, it turns out that these problems are actually in NP (see Section~\ref{sec:iBGP_CC}), and the proof of this fact naturally gives rise to a more general network design problem that we call \emph{Constrained Connectivity}.  In this problem we are given a graph $G=(V,E)$ and for each pair of nodes $(u,v)\in V\times V$ we are given a set $S(u,v)\subseteq V$.
Each such $S(u,v)$ is called a {\em safe set} and it is assumed that
$u,v \in S(u,v)$.  We say that a subgraph $H=(V,F)$
of $G$ is {\em safely connected} if for each pair of nodes $(u,v)$ there
is a path in $H$ from $u$ to $v$ where each node in the path is in $S(u,v)$.  As with iBGP, we are interested in two optimization versions of this problem:
\begin{enumerate}
\item {\sc Constrained Connectivity-Sum}: compute a safely
connected subgraph $H$ with the minimum number of edges, and
\item {\sc Constrained Connectivity-Degree}: compute a safely
connected subgraph $H$ that minimizes the maximum degree over all nodes.
\end{enumerate}

It turns out (see Theorem~\ref{thm:iBGP_safe_sets}) that the iBGP problems can be viewed as Constrained Connectivity problems with $G = K_n$ and safe sets defined in a particular geometric way.  While the motivation for studying Constrained Connectivity comes from iBGP, we believe that it is an interesting problem in its own right.  It is an extremely natural and general network design problem that, somewhat surprisingly, seems to have not been considered before.  While we only provide negative results for the general problem (hardness of approximation and integrality gaps), a better understanding of Constrained Connectivity might lead to a better understanding of other network design problems, both explicitly via reductions and implicitly through techniques.  For example, many of the techniques used in this paper come from recent literature on directed spanners~\cite{BGJRW09,DK11,BBMRY11}, and given these similarities it is not unreasonable to think that insight into Constrained Connectivity might provide insight into directed spanners.

For a more direct example, there is a natural security application of Constrained Connectivity.  Suppose we have $n$ players who wish to communicate with each other
but they do not all trust one another with messages they send to others.
That is, when $u$ wishes to send a message to $v$ there is a subset
$S(u,v)$ of players that it trusts to see the messages that it sends to $v$.
Of course, if for every pair of players there were direct communication
channels between the two players, then there would be no problem.
But suppose there is a cost to protect communication channels from
eavesdropping or other such attacks.
Then a goal would be to have a network of fewer than $O(n^2)$
communication channels that would still allow a route from
each $u$ to each $v$ with the route completely contained within $S(u,v)$.
Thus this problem defines a {\sc Constrained Connectivity-Sum} problem.

\subsection{Summary of Main Results}

In Section \ref{sec:CC_Kn} we
give a polynomial approximation for the iBGP problems, by giving the same approximations for the more general problem of Constrained Connectivity on $K_n$.

\vspace{.1in}
{\bf Theorem~\ref{thm:approx}.}\hspace{.01in}
  There is an $\tilde{O}(\sqrt{n})$-approximation to the Constrained
  Connectivity problems on $K_n$.
\vspace{.1in}

{\bf Corollary.}\hspace{.01in}
  There is an $\tilde{O}(\sqrt{n})$-approximation to {\sc iBGP-Sum} and
  {\sc iBGP-Degree}.
\vspace{.1in}

To go along with these theoretical upper bounds, we design a different (but related) algorithm for {\sc Constrained Connectivity-Sum} on $K_n$ that provides a worse theoretical upper bound (a $\tilde{O}(n^{2/3})$-approximation) but is faster in both practice and theory, and show by simulation on five real AS topologies (Telstra, Sprint, NTT, TINET, and Level 3) that in practice it provides an extremely good approximation.  Details of these simulations are in Section~\ref{sec:Simulations}

To complement these upper bounds, in Section~\ref{sec:iBGP_hardness} we show that the iBGP problem is hard to approximate, even with the extra power afforded us by the geometry of the safe sets:
\vspace{.1in}

{\bf Theorems~\ref{thm:iBGP_deghard} and~\ref{thm:iBGP_sumhard}.}\hspace{.01in}
  It is NP-hard to approximate {\sc iBGP-Sum} or {\sc iBGP-Degree} to
  a factor better than $\Omega(\log n)$.
\vspace{.1in}

We then study the more general Constrained Connectivity problems, and in Section~\ref{sec:CC} we show that the fully general
constrained connectivity problems are hard to approximate:

\vspace{.1in}
{\bf Theorem~\ref{thm:ccsumhard}.}\hspace{.01in}
  {\sc The Constrained Connectivity-Sum} and \textsc{Constrained
    Connectivity-Degree} problems do not admit a $2^{\log^{1-\epsilon}
    n}$-approximation algorithm for any constant $\epsilon > 0$ unless
  $\text{NP} \subseteq \text{DTIME}(n^{\polylog(n)})$
\vspace{.1in}

This is basically the same inapproximability factor as for Label
Cover, and in fact our reduction is from a minimization version of
Label Cover known as {\sc Min-Rep}.  Moreover, we show that the
natural LP relaxation has a polynomial integrality gap of
$\Omega(n^{\frac13 - \epsilon})$.

Finally, in Section~\ref{sec:hierarchical}
 we consider some other special cases of Constrained
Connectivity that turn out to be easier.
In particular, we say that a collection of
safe sets is \emph{symmetric} if $S(x,y) = S(y,x)$ for all $x,y \in V$ and that it is \emph{hierarchical} if for all $x,y,z \in V$, if $z \in
S(x,y)$ then $S(x,z) \subseteq S(x,y)$ and $S(z,y) \subseteq S(x,y)$.  It turns out that all of our hardness results and integrality gaps also hold for symmetric instances, but adding the hierarchical property makes things easier:
\vspace{.1in}

{\bf Theorem~\ref{lem:hierarchy_main}.}\hspace{.01in}
  {\sc Constrained Connectivity-Sum} with symmetric and hierarchical safe sets can
  be solved optimally in polynomial time.

\subsection{Related Work}
Issues involving eBGP, the version of BGP that routers in different ASes use
to announce routes to one another, have recently received
significant attention from the theoretical computer science community,
especially stability and game-theoretic issues
(e.g., ~\cite{griffin:02,levin:08,fabrikant:08}).
However, not nearly as much work has
been done on problems related to iBGP which distributes routes
internally in an AS.
There has been some work on the problem of guaranteeing
hot-potato routing in any AS with a route reflector
architecture~\cite{rfc:2796}.
These earlier papers did not consider the issue of
finding small signaling graphs that achieved the hot-potato goal.
Instead they either provided sufficient conditions for correctness
relating the underlying physical network with the route
reflector configuration~\cite{griffinwilfong:2002a} or they showed
that by allowing some specific extra routes to be announced (rather
than just the one chosen route) they would guarantee a version of
hot-potato routing~\cite{basu}.
The first people to consider the problem of
designing small iBGP overlays subject to achieving hot-potato
correctness were Vutukuru et al.~\cite{vutukuru:06}, who used graph
graph partitioning schemes to give such configurations.  But while
they proved that their algorithm gave correct configurations, they
only gave simulated evidence that the configurations it produced were
small.  Buob et
al.~\cite{buob:08} considered the problem of designing small correct
solutions and gave a mathematical programming formulation, but
then simply solved the integer program user super-polynomial time
algorithms.

\section{Preliminaries}

\subsection{Relationship between iBGP and Constrained Connectivity} \label{sec:iBGP_CC}

We will now show that the iBGP problems are just special
cases of {\sc Constrained Connectivity-Sum} and
{\sc Constrained Connectivity-Degree}.  This will be a natural consequence of the proof that {\sc iBGP-Sum} and {\sc iBGP-Degree} are in NP.

%\subsubsection{Safe Set Definitions for iBGP Problems}
%\label{sec:safeset}

  To see this we will need the following definitions.  We
will assume that there are no ties, i.e.\ all distances are distinct.
For two routers $x$ and $y$, let $D(x,y) = \{w : d(x,w) > d(x,y)\}$ be
the set of routers that are farther from $x$ than $y$ is.  Let $S(x,y)
= \{w : d(w,y) < d(w, D(x,y))\} \cup \{y\}$ be the set of
routers that are closer to $y$ than to any router not in the ball
around $x$ of radius $d(x,y)$.  We will refer to $S(x,y)$ as ``safe"
routers for the pair $(x,y)$.
A path from $x$ to $y$ in a signaling graph is said to be a
{\em safe signaling path} if it is contained in $S(x,y)$.
It turns out that these safe sets characterize correct signaling graphs:

\begin{theorem} \label{thm:iBGP_safe_sets} An iBGP signaling graph $H$
  is correct if and only if for every pair $(x,y) \in V \times V$
  there is a signaling path from $y$ to $x$ that uses only routers in
  $S(x,y)$.
\end{theorem}
\begin{proof}
  We first show that if every pair has a safe signaling path then every
  node hears about the route that has the
  closest egress router no matter what the set of
  egress routers $X_F$ is.  This is simple: let $x$ be a router, and let $y$ be its
  closest egress router. Let $r$ be the route whose egress router is $y$.
  By assumption there is a signaling path from $y$ to $x$
  that uses only routers in $S(x,y)$.  By definition, every one of these
  routers is closer to $y$ than to any router farther from $x$ than $y$
  is.  Since $y$ is the closest egress to $x$, this means that for all of
  the routers in $S(x,y)$, $y$ will be the closest egress router.
  A simple induction then shows that the routers in a safe signaling
  path will each choose $r$ and hence tell their iBGP neighbor in the
  path about $r$.
  That is, $x$ hears about $r$.

  For the other direction we need to show that if a signaling graph
  is correct then every pair has a safe signaling path.  For
  contradiction, suppose that there is no safe signaling path from $y$
  to $x$.  Let $X_F$, the set of egress routers, be $D(x,y)
  \cup \{y\}$.  Let $r$ be the route whose egress router is $y$.
  Since every router in $D(x,y)$ is farther from $x$ than
  $y$ is, this means that for this set of egress routers $x$ is closer to $y$
  than any other egress.  By correctness we know that $x$ does hear about
  $y$.  Let $y = a_1, a_2, \dots, a_k = x$ be the (or at least a)
  signaling path from $y$ to $x$ through which $x$ hears about $r$.
  Since there are no safe signaling paths from $y$ to $x$, we know that
  there exists some $i$ such that $a_i \not\in S(x,y)$.  This means that
  there is some $w \in D(x,y)$ such that $d(a_i, w) < d(a_i, y)$.  Since
  we assumed correctness we know that $a_i$ heard about the route with
  the closest egress router $z$ to $a_i$, and $z\not= y$
  (since $w$ in particular is closer).
  So $a_i$ will not tell its iBGP neighbors about $r$,
  which is a contradiction since $a_i$ is
  on the signaling path from which $x$ heard about $r$.  Thus a safe
  signaling path must exist.
\end{proof}

Note that this condition is easy to check in polynomial time, so we
have shown membership in NP.
Also this characterization shows that the problems
{\sc iBGP-Sum} and {\sc iBGP-Degree} are Constrained Connectivity
problems where the underlying graph $G$ is $K_n$ and the safe sets are defined by certain geometric properties.  While the proof of this is obviously relatively simple, we believe that it is an important contribution of this paper as it allows us to characterize the behavior of a protocol (iBGP) using only the static information of the signaling graph and the network distances.

\subsection{Linear Programming Relaxations}

There are two obvious linear programming relaxations of the {\sc
  Constrained Connectivity} problems (and thus the iBGP problems): the \emph{flow LP} and the
\emph{cut LP}.  For every pair $(u,v) \in V \times V$ let
$\mathcal{P}_{uv}$ be the collection of $u-v$ paths that are contained
in $S(u,v)$.  The flow LP has a variable $c_e$ for every edge $e \in
E$ (called the \emph{capacity} of edge $e$) and a variable $f(P)$ for
every $u-v$ path in $\mathcal{P}_{uv}$ for every $(u,v) \in V \times
V$ (called the \emph{flow} assigned to path $P$).  The flow LP simply
requires that at least one unit of flow is sent between all pairs
while obeying capacity constraints:
\begin{align*}
\min &  \textstyle\sum_e c_e \\
\text{s.t} &\textstyle \sum_{P \in \mathcal{P}_{uv}} f(P) \geq 1 & \forall (u,v) \in V \times V \\
&\textstyle \sum_{P \in \mathcal{P}_{uv} : e \in P} f(P) \leq c_e & \forall e \in E, (u,v) \in V \times V \\
& 0 \leq c_e \leq 1 & \forall e \in E \\
& 0 \leq f(P) \leq 1 & \forall (u,v) \in V \times V, P \in \mathcal{P}_{uv}
\end{align*}

 This is obviously a valid relaxation of {\sc
  Constrained Connectivity-Sum}: given a valid solution to
{\sc Constrained Connectivity-Sum}, let $P_{uv}$ denote the required
safe $u-v$ path for every $(u,v) \in V \times V$.  For every edge $e$
in some $P_{uv}$ set $c_e$ to $1$, and set $f(P_{uv})$ to $1$ for
every $(u,v) \in V \times V$.  This is clearly a valid
solution to the linear program with the exact same value.  To change
the LP for {\sc Constrained Connectivity-Degree} we can just introduce
a new variable $\lambda$, change the objective function to $\min
\lambda$, and add the extra constraints $\sum_{v: \{v,u\} \in E} c_{\{u,v\}}
\leq \lambda$ for all $u \in V$.  And while this LP can be exponential in size (since there is a variable for every path), it is also easy to design a compact representation that has only $O(n^4)$ variables and constraints.  This compact representation has variables $f_{(u,v)}^{(x,y)}$ instead of $f(P)$, where $f_{(u,v)}^{(x,y)}$ represents the amount of flow from $u$ to $v$ along edge $\{u,v\}$ for the demand $(x,y)$.    Then we can write the normal flow conservation and capacity constraints for every demand $(x,y)$ independently, restricted to $S(x,y)$.  Indeed, this compact representation is one of the main reasons to prefer the flow LP over the cut LP.

The cut LP is basically equivalent to the flow LP, except that instead of requiring flow to be sent, it requires the min-cut to be large large enough. Given a pair $(u,v) \in V \times V$, let
$\mathcal{S}(u,v) = \{S \subset S(u,v) : u \in S \land v \not\in S\}$
be the collection of safe set cuts that separate $u$ and $v$.
Furthermore, given a set $S \in \mathcal{S}(u,v)$ let $\delta_{uv}(S)
= \{e \in {V \choose 2} : e \in (S, S(u,v) \setminus S)\}$ be the set
of safe edges that cross $S$.  The cut LP has a variable $x_e$ for
every edge $e$ (equivalent to $c_e$ in the flow LP), and is quite simple:
\begin{align*}
\min &\sum_e x_e \\
\text{s.t. } &\sum_{e \in \delta_{uv}(S)} x_e \geq 1 \qquad
\forall u,v \in V, S \in \mathcal{S}(u,v)
\end{align*}

This LP simply minimizes the sum of the edge variables subject to the
constraint that for every cut between two nodes there must be at least
one safe edge crossing it.  While the flow LP and the cut LP are not technically duals of each other (since capacities are variables), it is easy to see from the max flow-min cut theorem that they do in fact describe the same polytope (with respect to the capacity variables).  Thus integrality gaps for one automatically hold for the other, as do approximations achieved by LP rounding.

\section{Algorithms for iBGP and Constrained Connectivity on $K_n$} \label{sec:CC_Kn}

\subsection{$\tilde{O}(\sqrt{n})$-approximation}
In this section we show that there is a $\tilde{O}(\sqrt{n})$-approximation algorithm for both Constrained Connectivity problems as long as the underlying graph is the complete graph $K_n$.  This algorithm is inspired by the recent progress on directed spanners by Bhattacharyya et al.~\cite{BGJRW09}, Dinitz and Krauthgamer~\cite{DK11}, and Berman et al.~\cite{BBMRY11}.  In particular, we use the same two-component framework that they do: a randomized rounding of the LP and a separate random tree-sampling step.  The randomized rounding we do is simple independent rounding with inflated probabilities.  The next lemma implies that this works well when the safe sets are small.

\begin{lemma} \label{lem:sample}
Let $E' \subseteq E$ be obtained by adding every edge $e \in E$  to $E'$ independently with probability at least $\min\{12c_e \cdot |S(x,y)| \ln n, 1\}$.  Then with probability at least $1-1/n^3$, $E'$ will have a path between $x$ and $y$ contained in $S(x,y)$.
\end{lemma}
\begin{proof}
Let $(X,Y)$ be a partition of $S(x,y)$ so that $x \in X$ and $y \in Y$, i.e.~$(X,Y)$ is an $x-y$ cut of $S(x,y)$.  Note that there are only $2^{|S(x,y)|}$ such cuts, and by standard arguments if at least one edge from every cut is chosen to be in $E'$ then $E'$ contains an $x-y$ path in $S(x,y)$.  Since in any LP solution at least one unit of flow is sent from $x$ to $y$ in $S(x,y)$, every cut has capacity at least $1$.  Let $\delta(X,Y)$ be the set of edges that cross the cut $(X,Y)$.  If $c_e \geq 1/ (12 |N(x,y)| \ln n)$ for any $e \in \delta(X,Y)$ then $e$ is selected with probability $1$, and thus $(X,Y)$ is spanned.  Otherwise, the probability that no edge from $\delta(X,Y)$ is chosen is at most $\prod_{e \in \delta(X,Y)} (1-12 c_e \cdot |S(x,y)|\ln n) \leq \exp(-12|S(x,y)|\ln n \cdot \sum_{e \in \delta(X,Y)} c_e) \leq e^{-3|S(x,y)| \ln n}$.  Thus by a simple union bound the probability that we fail on \emph{any} cut is at most $2^{|S(x,y)|} e^{-12|S(x,y)| \ln n} \leq (2/e)^{-12|S(x,y)| \ln n} \leq 1/n^3$
\end{proof}

Another important part of our algorithm will be random sampling that is independent of the LP.  We
will use two different types of sampling: star sampling for the sum version and edge sampling for the degree version.  First we
consider star sampling, in which we independently sample nodes with
probability $p$, and every sampled node becomes the center of a star that spans the vertex set.

\begin{lemma} \label{lem:star_sample}
  All pairs with safe sets of size at least $s$ will be satisfied by
  random star sampling with high probability if $p = 3\ln n / s$.
\end{lemma}
\begin{proof}
  Consider some pair $(x,y)$ with $|S(x,y)| \geq s$.  If some node
  (say $z$) from $S(x,y)$ is sampled then the pair is satisfied, since
  the creation of a star at $z$ would create a path $x-z-y$ that would
  satisfy $\{x,y\}$.  The probability that no node from $S(x,y)$ is
  sampled is
  \begin{equation*}
    (1-p)^{|S(x,y)|} \leq (1-p)^s \leq e^{-ps} = e^{-3\ln n} = 1/n^3
  \end{equation*}
  Since there are less than $n^2$ pairs, we can take a union bound
  over all pairs $(x,y)$ with $|S(x,y)| \geq s$, giving us that all
  such pairs are satisfied with probability at least $1 - 1/n$.
\end{proof}

For edge sampling, we essentially consider the Erd\H{o}s-R\'{e}nyi
graph $G_{n,p}$, i.e.\ we just sample every edge independently with
probability $p$.  We will actually consider the union of $3 \log n$
independent $G_{n,p}$ graphs, where $p = \frac{(1+\epsilon)\log s}{s}$
for some small $\epsilon > 0$.  Let $H$ be this random graph.

\begin{lemma} \label{lem:edge_sample} With probability at least
  $1-1/n$, all pairs with safe sets of size at least $s$ will be
  connected by a safe path in $H$.
\end{lemma}
\begin{proof}
  Let $(x,y)$ be a pair with $|S(x,y)| \geq s$.  Obviously $(x,y)$ is
  satisfied if the graph induced on $S(x,y)$ is connected.  It is
  known \cite{bollobas:01} that there is some small $\epsilon$ with $0
  < \epsilon < 1$ so that $G_{s,p}$ is connected with probability at
  least $1/2$.  Since $H$ is the union of $3\log n$ instantiations of
  $G_{n,p}$, we know that the probability that the subgraph of $H$
  induced on $S(x,y)$ is not connected is at most $1/n^3$.  We can now
  take a union bound over all such $(x,y)$ pairs, giving us that the
  probability that there is some unsatisfied $(x,y)$ pairs with
  $|S(x,y)| \geq s$ is at most $1/n$.
\end{proof}

We will now combine the randomized rounding of the LP and the random sampling into a single approximation algorithm.  Our algorithm is
divided into two phases: first, we solve the LP and randomly include every edge $e$ with probability $O(c_e \sqrt{n} \ln n)$.  By Lemma~\ref{lem:sample} this takes care of safe sets of size at most $\sqrt{n}$.  Second, if the objective is to
minimize the number of edges we do star sampling with probability
$(3\ln n) / \sqrt{n}$, and if the objective is to minimize the maximum
degree we do edge sampling using the construction of Lemma
\ref{lem:edge_sample} with $s = \sqrt{n}$.  It is easy to see that this algorithm with high probability results in a valid solution that is a $\tilde{O}(\sqrt{n})$-approximation.

\begin{theorem} \label{thm:approx}
  This algorithm is a $\tilde{O}(\sqrt{n})$-approximation to both {\sc Constrained Connectivity-Sum} and {\sc Constrained Connectivity-Degree} on $K_n$.
\end{theorem}
\begin{proof}
  We first argue that the algorithm does indeed give a valid solution
  to the problem.  Let $(x,y)$ be an arbitrary pair.  If $|S(x,y)|
  \leq \sqrt{n}$, then Lemma \ref{lem:sample} implies that the first phase of
  the algorithm results in a safe path.  If $|S(x,y)| \geq \sqrt{n}$,
  then Lemma \ref{lem:star_sample} or Lemma \ref{lem:edge_sample}
  imply that the second phase of the algorithm results in a safe path.
  So every pair has a safe path, and thus the solution is valid.

  We now show that the cost of this algorithm is at most
  $\tilde{O}(\sqrt{n}) \times OPT$.  We first consider the objective
  function of minimizing the number of edges.  In the LP rounding step we only
  increase capacities by at most a factor of $\tilde{O}(\sqrt{n})$, so since the
  LP is a relaxation of the problem we know that the expected cost cost of
  the rounding is at most $\tilde{O}(\sqrt{n}) \times OPT$.  For phase 2, in
  expectation we chose $3\sqrt{n} \ln n$ stars, for a total of at most
  $3 n^{3/2} \ln n$ edges.  But since there is a demand for every pair
  we know that $OPT \geq n-1$, so phase 2 has total cost at most
  $\tilde{O}(\sqrt{n}) \times OPT$.

  If instead our objective function is to minimize the maximum degree,
  then since phase 1 only increases capacities by $\tilde{O}(\sqrt{n})$ we know
  that after phase 1 the maximum degree is at most $\tilde{O}(\sqrt{n}) \times
  OPT$ (by a Chernoff bound, with high probability every vertex has degree at most $\tilde{O}(\sqrt{n})$ times its fractional degree in the LP).  In phase 2, a simple Chernoff bound implies that with high
  probability every node gets $\tilde{O}(\sqrt{n})$ new edges, and thus
  the node with maximum degree still has degree at most
  $\tilde{O}(\sqrt{n}) \times OPT$.
\end{proof}

\subsection{Primal-Dual Algorithm} \label{sec:PD}

We also have a primal-dual algorithm that gives a slightly worse result for the {\sc Constrained Connectivity-Sum} problem.  While this algorithm and its analysis is slightly more complicated and only works for the Sum version, by not solving the linear program we get a faster algorithm.  In particular, the best known algorithms for solving linear programs with $m$ variables take $\Omega(m^{3.5})$ time on general LPs, so since there are $n^4$ variables in the compact version of the flow LP this takes $\Omega(n^{12.5})$ time.  The primal-dual algorithm, on the other hand, is significantly faster: a na\"{i}ve analysis shows that it takes $\tilde{O}(n^6)$ time.

In this algorithm we use the cut LP rather than the
flow LP (in fact, the algorithm is very similar to the primal-dual
algorithm for Steiner Forest, which uses a similar cut LP but doesn't
have to deal with safe sets).  Since this is a primal-dual algorithm,
instead of solving and rounding the cut LP we will consider the dual,
which has a variable $y_S^{uv}$ for every pair $(u,v)$ and $S \in
\mathcal{S}(u,v)$.  We say the an edge $e \in S(u,v)$ if both endpoints of $e$ are in $S(u,v)$.
\begin{align*}
\max &\sum_{u,v \in V} \quad \sum_{S \in \mathcal{S}(u,v)}
y_S^{uv} \\
\text{s.t. } &\sum_{u,v \in V : e \in S(u,v)}\quad \sum_{S \in
  \mathcal{S}(u,v) : e \in \delta_{uv}(S)} y_S^{uv} \leq 1 \qquad
\forall e \in {V \choose 2}
\end{align*}

Unfortunately we will not be able to use a pure primal-dual
approximation, but will have to trade off with a random sampling
scheme as in the rounding algorithm.  So instead of this primal, we
will only have constraints for $u,v \in V$ with $|S(u,v)| \leq t$ for
some parameter $t$ that we will set later.  Thus in the dual we will
only have variables $y_S^{uv}$  for $(u,v)$ with $|S(u,v)| \leq t$.
This clearly preserves the property that the primal is a valid
relaxation of the actual problem.  Let $D = \{(u,v) : |S(u,v)| \leq
t\}$.

Our primal-dual algorithm, like most primal-dual algorithms,
maintains a set of \emph{active} dual variables that it increases
until some dual constraint becomes tight.  Once that happens we buy an
edge (i.e.\ set some $x_e$ to $1$ in the primal), change the set of
active dual variables, and repeat.  We do this until we have a
feasible primal.

Initially our primal solution $H$ is empty and the active dual
variables are $y_{\{u\}}^{uv}$ for every $(u,v) \in D$, i.e.\ every
node $u$ has an active dual variable for every other $v$ that it has a
demand with corresponding to the cut in $S(u,v)$ that is the singleton
$\{u\}$.  We raise these variables uniformly until some constraint
(say the one for $e = \{w,z\}$) becomes tight.  At this point we add
$e$ to our current primal solution $H$.  We now change the active dual
variables by ``merging'' moats that cross $e$.  In particular, there are some active variables $\{y_S^{uv}\}$ where $e \in \delta_{uv}(S)$ (which implies
that $w,z \in S(u,v)$ as well).  Let $H|_{S(u,v)}$
denote the subgraph of $H$ induced on
$S(u,v)$.  Without loss of generality we can assume that $w \in S$ and
$z \not\in S$.  Let $T \subset S(u,v)$ be the connected component of
$H|_{S(u,v)}$ containing $z$.  We now make $y_S^{uv}$ inactive, and
make $y_{S \cup T}^{uv}$ active.  We do this for all such active
variables, and then repeat this process (incrementing all dual
variables until some dual constraint becomes tight, adding that edge
to $H$, and then merging moats that cross it) until all pairs $(u,v)
\in D$ have a safe path in $H$.

\begin{lemma} \label{lem:PD_feasible} This algorithm always maintains
  a feasible dual solution and an active set that does not contribute
  to any tight constraint.
\end{lemma}
\begin{proof}
  We will show this by induction, where the inductive hypothesis is
  that the dual solution is feasible and that no dual variables that
  contribute to a tight constraint are active.  Initially all dual
  variable are $0$, so it is obviously a feasible solution and no
  constraints are tight.  Now suppose this is true after we add some
  edge $e'$.  We need to show that it is also true after we add the
  next edge $e = \{w,z\}$.  By induction the dual solution after we
  added $e'$ is feasible and none of the active dual variables
  contribute to any tight constraints.  Thus raising the active dual
  variables until some constraint becomes tight maintains dual
  feasibility.

  To prove that no active variables contribute to a tight constraint,
  note that the only new tight constraint is the one corresponding to
  $e$.  The only variables contributing to that constraint are of the
  form $y_S^{uv}$ where $e \in \delta_{uv}(S)$.  But our algorithm
  made all of these variables inactive, and only added new active
  variables for sets $S'$ that contain both $w$ and $z$ and thus do
  not contribute to the newly tight constraint.  Furthermore, these
  sets $S'$ are formed by the union of $S$ and the connected component
  in $H|_{S(u,v)}$ containing the other endpoint, so no newly active
  variable contributes to a constraints that became tight previously
  (since they correspond to edges in $H$).
\end{proof}

\begin{theorem} \label{thm:PD_main}
  The primal-dual algorithm returns a graph $H$ with at most $O(t^2)
  \times OPT$ edges in which every pair $(u,v)$ with $|S(u,v)| \leq t$
  has a safe path.
\end{theorem}
\begin{proof}
  After every iteration of the algorithm all of the tight constraints
  are added to $H$, which together with Lemma \ref{lem:PD_feasible}
  implies that the algorithm never gets stuck.  Thus it will run until
  every pair $u,v$ with $|S(u,v)| \leq t$ has a safe path.  It just
  remains to show that the total number of edges returned is at most
  $O(t^2) \times OPT$.  To see this, note that every edge in $H$
  corresponds to a tight constraint in the feasible dual solution we
  constructed, so if $e \in H$ then $\sum_{u,v : e \in S(u,v)} \sum_{S
    \in \mathcal{S}(u,v) : e \in \delta_{uv}(S)} y_S^{uv} = 1$.  Thus
  we have that
  \begin{align*}
    |H| &= \sum_{e \in H} 1 = \sum_{e \in H} \sum_{(u,v) \in D : e \in S(u,v)}
    \sum_{S \in \mathcal{S}(u,v) : e \in \delta_{uv}(S)} y_S^{uv} \\
    &= \sum_{(u,v) \in D} \sum_{S \in \mathcal{S}(u,v)} \sum_{e \in
      \delta_{uv}(S) \cap H} y_S^{uv} \\
    &= \sum_{(u,v) \in D} \sum_{S \in \mathcal{S}(u,v)} |H \cap
    \delta_{uv}(S)| y_S^{uv} \\
    & \leq t^2 \sum_{(u,v) \in D} \sum_{S \in \mathcal{S}(u,v)} y_S^{uv}
    \\
    & \leq t^2 \times OPT
  \end{align*}
  where the last inequality is by duality, and the next to last
  inequality is because $|H \cap \delta_{uv}(S)| \leq
  {|\delta_{uv}(S)| \choose 2} \leq t^2$ (since $(u,v) \in
  D$).
\end{proof}

\begin{lemma} \label{lem:PD_time}
The primal-dual algorithm takes at most $\tilde{O}(n^6)$ time.
\end{lemma}
\begin{proof}
The primal-dual algorithm adds at least one new edge per iteration, so there can be at most $n^2$ iterations.  In each iteration we have to figure out the current value of every dual constraint and the number of active variables in each constraint, which together will imply what the next tight constraint is and how much to raise the $y$ variables.  We then need to raise the active variables by that amount and merge moats.  Note that for every demand there are at most two active moats, so the total number of active variables is at most $O(n^2)$.  Thus each iteration can be done in time $O(n^4)$, where the dominant term is the time taken to calculate the value of each dual constraint.  So the total time is $\tilde{O}(n^6)$, where there are extra poylogarithmic terms due to data structure overhead.
\end{proof}

Now we can trade this off with the random sampling solution for large
safe sets to get an actual approximation algorithm:

\begin{theorem} \label{thm:PD_total}
  There is a $\tilde{O}(n^{2/3})$ approximation algorithm for the
  {\sc Constrained Connectivity-Sum} problem on $K_n$ that runs in time $\tilde{O}(n^6)$.
\end{theorem}
\begin{proof}
Our algorithm first runs the primal-dual algorithm with $t = O((n \log n)^{1/3})$.  By Theorem~\ref{thm:PD_main}, this returns a graph $H$ with at most $O((n \log n)^{2/3}) \times OPT$ edges in which there is a safe path for every $(u,v)$ with $|S(u,v)| \leq O((n \log n)^{1/3})$.  We then use the random star sampling of Lemma~\ref{lem:star_sample} with $s = \Omega((n \log n)^{1/3})$ and thus $p = O((\log n)^{1/3} / n^{2/3})$.  By Lemma~\ref{lem:star_sample} this satisfies the rest of the demands (the pairs $(u,v)$ with $|S(u,v)| \geq s$) with high probability, and the number of edges added is with high probability at most $O(p n^2) = O((n \log n)^{2/3} n) = O((n \log n)^{2/3}) \times OPT$ as desired.

The time bound follows from Lemma~\ref{lem:PD_time} together with the trivial fact that star sampling can be done in $O(n^6)$ time.
\end{proof}

\subsection{Simulations} \label{sec:Simulations}

In this section we discuss some the results of simulations using our
algorithms.  While we believe that the main contribution of this work
is theoretical, it is interesting that the algorithms are fast enough
to be practical and give solutions that are in practice far superior
to the worst case $\tilde{O}(n^{2/3})$ bound.

We implemented both the LP rounding and the primal-dual algorithm for
the {\sc iBGP-Sum} problem.  However, the rounding algorithm turned
out to be impractical, mainly due to memory constraints.  Recall that
in the compact version of the flow LP there is a flow variable
$f_{(u,v)}^{(x,y)}$ for every pair $(u,v)$ and $(x,y)$.  This variable
denotes the amount of flow from $u$ to $v$ along the edge $\{u,v\}$ for
the demand $(x,y)$.  There are also $\Theta(n^4)$ capacity
constraints.  So on even a modest size AS topology, say one with $50$
nodes, the linear program has over six million variables and
constraints.  Running on a commodity desktop, the memory used by CPLEX
merely to create and store this LP results in an extremely large
running time, even without attempting to solve it.  Our primal-dual
algorithm, on the other hand, only needs to keep track of $O(n^2)$
active dual variables and the current values of the $O(n^2)$ dual
constraints.  So we can actually run this algorithm on reasonably
sized graphs.

One change that we make from the theoretical algorithm is the tradeoff
with random sampling.  In the theoretical analysis we are only able to
get a nontrivial approximation bound by using the primal-dual
algorithm to handle small safe sets and random sampling to handle
large safe sets, but experimentation revealed that the simpler
algorithm of using the primal-dual technique to handle all safe sets
was sufficient.

\begin{table}
\begin{center}
\begin{tabular}{|c|l|c|c|}
\hline
AS & Name & Number of PoPs & Number of links\\
\hline
1221 & Telstra & 44 & 88 \\
1239 & Sprint & 52 & 168 \\
2914 & NTT & 70 & 222 \\
3257 & TINET & 41 & 174 \\
3356 & Level 3 & 63 & 570 \\
\hline
\end{tabular}
\end{center}
\caption{ISP Topologies Used}
\label{tab:topologies}
%\vspace*{-.4in}
\end{table}

To test out this algorithm we ran it on five real-world ISP topologies
with link weights given by the Rocketfuel project~\cite{Rocketfuel04}.
Our implementation is still relatively slow, so we consider
Point-of-Presence level topologies rather than router-level
topologies.  We feel that this is not unrealistic, though, since in
practice the routers at a given PoP would probably just use a single
router at that PoP as a route reflector~\cite[Section
3.1]{POAMZ10}. The topologies we used are summarized in Table
\ref{tab:topologies}.

We compare the number of iBGP sessions used by a full mesh to the
number of edges in the overlay produced by the primal-dual algorithm.
We assume (conservatively) that all the nodes in the topology are
external BGP routers.  Our results are shown in Table
\ref{tab:results} and in Figure \ref{fig:main}.  These results show
that the primal-dual algorithm gives graphs that are much smaller than
the default full mesh.  Of course, we do not model additional
requirements such as fault-tolerance and stability, but the massive
gap suggests that even if adding extra requirements results in
doubling or tripling the size of the overlay we will still see a large
benefit over the full mesh.  Moreover, these results show that the
$\tilde{O}(n^{2/3})$ upper bound on the approximation ratio that we
proved in Section~\ref{sec:PD} is extremely pessimistic.  On these
actual topologies the primal-dual algorithm gives results that are
only slightly larger than $n$ (the worst case is for Level 3, in which
the primal-dual algorithm gives an overlay with about $2.75 \times n$
edges).  Since $n-1$ is an obvious lower bound (the overlays clearly
must be connected), this means that in practice our algorithm gives a
$O(1)$-approximation.

\begin{table}
\begin{center}
\begin{tabular}{|c|c|c|c|}
\hline
AS & full-mesh & Primal-Dual & Fraction of full-mesh \\
\hline
1221 & 946 & 44 & 4.65\% \\
1239 & 1326 & 83 & 6.26\% \\
2914 & 2415 & 109 & 4.5\% \\
3257 & 820 & 75 & 9.15\% \\
3356 & 1953 & 173 & 8.86\% \\
\hline
\end{tabular}
\end{center}
\caption{Primal-Dual vs.\ full-mesh}
\label{tab:results}
%\vspace*{-.2in}
\end{table}

\begin{figure}
\begin{center}
\includegraphics[scale=0.65]{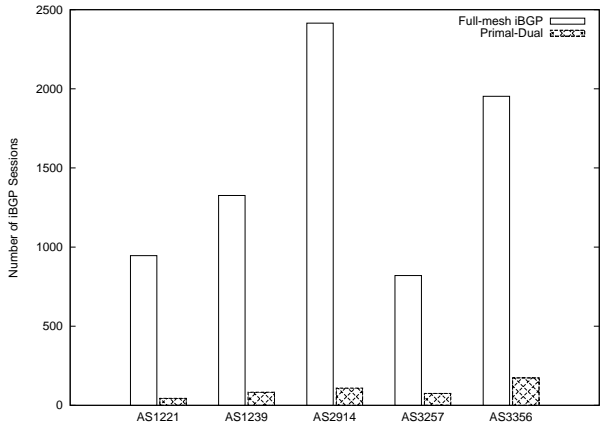}
\end{center}
\caption{Primal-Dual vs.\ full-mesh}
\label{fig:main}
\vspace*{-.2in}
\end{figure}

\section{Complexity of {\sc iBGP-Sum} and {\sc iBGP-Degree}} \label{sec:iBGP_hardness}

In this section we will show that the iBGP problems are $\Omega(\log
n)$-hard to approximate by a reduction from {\sc Hitting Set} (or
equivalently from {\sc Set Cover}).  This is a much weaker hardness
than the $2^{\log^{1-\epsilon} n}$ hardness that we prove for the general
Constrained Connectivity problems in Section~\ref{sec:CC}, but the iBGP problems are much more
restrictive.  We note that this $\Omega(\log n)$ hardness is easy to prove for Constrained Connectivity on $K_n$; the main difficulty is constructing a metric so that the geometrically defined safe sets of iBGP have the structure that we want.

We begin by giving a useful gadget that
encodes a {\sc Hitting Set} instance as an instance of an iBGP problem
in which all we care about is minimizing the degree of a particular
vertex.  We will then show how a simple combination of these gadgets
can be used to prove that {\sc iBGP-Degree} is hard to approximate,
and how more complicated modifications to the gadget can be used to
prove that {\sc iBGP-Sum} is hard to approximate.

Suppose we are given an instance of hitting set with elements $1,2,
\dots, n$ (note that we are overloading these as both integers and
elements) and sets $T_1, T_2, \dots, T_{m}$.  Our gadget will contain
a node $x$ whose degree we want to minimize, a node $a_i$ for all
elements $i \in \{1,\dots, n\}$, and a node $b_{T_j}$ for each set
$T_j$ in the instance.  We will also have four extra ``dummy" nodes:
$z,y,u$, and $h$.  The following table specifies some of the distances
between points.  All other distances are the shortest path graph
distances given these. Let $M$ be some large value (e.g.~$20$), and let $\epsilon$
be some extremely small value larger than $0$.

\begin{center}
\begin{tabular}{c|ccccccc}
 & x & z & y & $a_i$ & $b_{T_j}$ & u & h \\
\hline
x &  & M &  &  & $M+ 1.4 + j\epsilon$ & &  \\
z & M & &  $1.5$ & $1+i\epsilon$ &  & $2$ &  \\
y & & $1.5$ & & \\
$a_i$ & & $1+i\epsilon$& & & $1+(i+j)\epsilon$ (if $i \in T_j$) & $1.1$&  \\
$b_{T_j}$ & $M +1.4 + j\epsilon$ & & & $1+(i+j)\epsilon$ (if $i \in T_j$) & & & $1+j\epsilon$ \\
$u$ & & $2$ & & $1.1$\\
$h$ &  & & & & $1+j\epsilon$
\end{tabular}
\end{center}

It is easy to check that this is indeed a metric space.  Informally,
we want to claim that any solution to the iBGP problems on this
instance must have an edge from $x$ to $a_i$ nodes such that the
associated elements $i$ form a hitting set.  Here $y,u$, and $h$ are nodes that force the safe sets into the form we want, and $z$ is used to guarantee the existence of a small solution.

\begin{lemma} \label{lem:gadget_necessary} Let $E$ be any feasible
  solution to the above iBGP instance.  For every vertex $b_{T_j}$
  there is either an edge $\{x, b_{T_j}\} \in E$ or an edge $\{x,
  a_i\} \in E$ where $i \in T_j$
\end{lemma}
\begin{proof}
  We will prove this by analyzing $S(x,b_{T_j})$.  If we can show that
  $S(x, b_{T_j}) = \{x, b_{T_j}\} \cup \{a_i : i \in T_j\}$ then we
  will be finished.  Note that $d(x, b_{T_j}) = M + 1.4 + j\epsilon$,
  so the vertices outside $B(x, d(x, b_{T_j}))$ are $y$ (distance
  $M+1.5$ from $x$), $u$ (distance $M+2$ from $x$), $h$ (distance at least
  $M+2.4$ from $x$), and $b_{T_k}$ with $k > j$ (distance
  $M+1.4+k\epsilon$ from $x$).  The vertices inside the ball are
  $x,z$, all $a_i$ nodes, and $b_{T_k}$ with $k \leq j$.

Obviously $x$ and $b_{T_j}$ are in $S(x, b_{T_j})$ by definition.  Let
$a_i$ be a vertex with $i \in T_j$.  It is easy to verify that $a_i$
is closer to $b_{T_j}$ than to any vertex outside of the ball: it has
distance $1+ (i+j)\epsilon$ from $b_{T_j}$, distance $1+(i+k)\epsilon$
from $b_{T_k}$ with $k > j$, distance $2.5 + i\epsilon$ from $y$,
distance $1.1$ from $u$, and distance greater than $2$ from $h$.  So
$a_i \in S(x, b_{T_j})$ as required.  On the other hand, suppose $i
\not\in T_j$.  Then $d(a_i, b_{T_j}) > 2$, while $d(a_i, u) = 1.1$, so
$a_i \not\in S(x, b_{T_j})$.  Similarly, any vertex $b_{T_k}$ with $k
< j$ is closer to $h$ (distance $1+j\epsilon$) than to $b_{T_j}$
(distance at least $2$) and $z$ is closer to $y$ (distance $1.5$) than
to $b_{T_j}$ (distance at least $2$).  Thus $S(x, b_{T_j}) = \{x,
b_{T_j}\} \cup \{a_i : i \in T_j\}$, so $E$ must include an edge from
$x$ to either $b_{T_j}$ or an $a_i$ with $i \in T_j$.
\end{proof}

We now want to use this gadget to prove logarithmic hardness for {\sc
  iBGP-Sum}.  We will use the basic gadget but will duplicate $x$.  So there
will be $\ell$ copies of $x$, which we will call $x_1, x_2, \dots,
x_{\ell}$, and their distances are defined to be $d(x_i, z) = M +
i\epsilon$ and $d(x_i, b_{T_j}) = M + 1.4 + (i+j)\epsilon$ with all
other distances defined to be the shortest path.  Note that all we did
was modify the gadget to ``break ties'' between the $x_i$'s.  Also
note that the shortest path between $x_i$ and $x_j$ is through $z$, for
a total distance of $2M + (i+j)\epsilon$.  As before, let $H$ be the
smallest hitting set.

\begin{lemma} \label{lem:ibgp-sum-necessary}
  Any feasible {\sc iBGP-Sum} solution has at least $\ell |H|$ edges.
\end{lemma}
\begin{proof}
  It is easy to see that Lemma \ref{lem:gadget_necessary} still holds,
  i.e.~that $S(x_i, b_{T_j}) = \{x_i, b_{T_j}\} \cup \{a_k : k \in
  T_j\}$.  Intuitively this is because all other $x$ nodes are outside
  of $B(x_i, d(x_i b_{T_j}))$ and all distances from $x$ to the gadget
  are the same as before except with an additional $i\epsilon$.  This
  implies that the number of $a_k$ and $b_{T_j}$ nodes adjacent to
  $x_i$ in any feasible solution must be at least $|H|$, since if
  there were fewer such adjacent nodes it would imply the existence of
  a smaller hitting set (any $b_{T_j}$ nodes adjacent to $x_i$ could
  just be covered using an arbitrary element in $T_j$ at the same cost
  as using the set itself).  Thus the total number of edges must be at
  least $\ell |H|$.
\end{proof}

\begin{lemma} \label{lem:ibgp-sum-sufficient}
  There is a feasible {\sc iBGP-Sum} solution with at most $\ell |H| +
  \ell + (m+n+4)^2$ edges.
\end{lemma}
\begin{proof}
  The solution is simple: create a clique on the $a_i, b_{T_j},z,u,y,h$
  nodes (which obviously has size at most $(m+n+4)^2$), include an
  edge from every $x_i$ to $z$ (another $\ell$ edges) and include an
  edge from every $x_i$ to every $a_k$ with $k \in H$ (another $\ell
  |H|$ edges).  Obviously there are the right number of edges in this
  solution, so it remains to prove that it is feasible.  To show this
  we partition the pairs into types and show that every pair in every
  type is satisfied.  The types are
  \begin{inparaenum}[\itshape 1\upshape)]
  \item $x_i - b_{T_j}$,
  \item $x_i - h$,
  \item $x_i - x_j$,
  \item $x_i - \alpha$ (where $\alpha$ is any other node in the gadget
    not included in a previous type), and
  \item $\alpha - x_i$
  \end{inparaenum}
  This is clearly an exhaustive partitioning, so we can just
  demonstrate that each type is satisfied in turn.

  For the first type we already showed that $S(x_i, b_{T_j})$ includes
  all $a_k$ where $k \in T_j$.  Since $H$ is a valid hitting set $x_i$
  must be adjacent to one such $a_k$, which in turn is adjacent to
  $b_{T_j}$, forming a valid safe path.  For the second type the only
  vertices outside $B(x_i, d(x_i,h))$ are $x_j$ with $j \neq i$, and
  $z$ is closer to $h$ than to any such $x_j$.  Thus $z \in S(x_i, h)$
  so the path $x_i - z - h$ in our solution is a valid safe path.  For
  the third type the vertices outside $B(x_i, d(x_i, x_j))$ are $\{x_k
  : k > j \text{ and } k \neq i\}$.  Because of the tie-breaking we
  introduced, $d(z, x_j) = M + j\epsilon$ while $d(z, x_k) = M + k
  \epsilon > M + j\epsilon$, and thus $z \in S(x_i, x_j)$ and so the
  path $x_i - z - x_j$ in our solution is a valid safe path.  The
  fourth type is even simpler, since $\alpha$ must be either $z,u,y$,
  or an $a_k$ node and the shortest path from $x_i$ to any of these is
  through $z$.  So $z \in S(x_i, \alpha)$ and $x_i - z - \alpha$ is a
  valid safe path.  Finally, for the last type the vertices outside
  $B(\alpha, d(\alpha, x_i))$ are $\{x_k : k > i\}$, and $z$ is closer
  to $x_i$ (distance $M + i\epsilon$) than any such $x_k$ (distance
  $M+k\epsilon$).  So again $z \in S(\alpha, x_i)$ and thus $\alpha -
  z - x_i$ is a valid safe path.
\end{proof}

\begin{theorem}
\label{thm:iBGP_sumhard}
  It is NP-hard to approximate {\sc iBGP-Sum} to a factor better than
  $\Omega(\log N)$, where $N$ is the number of vertices in the
  metric.
\end{theorem}
\begin{proof}
  It is known that there is some $\beta$ for which it is NP-hard to
  distinguish hitting set instances with a hitting set of size at most
  $\beta$ from instances in which all hitting sets have size at least
  $\beta \ln m$.  In the first case we know from Lemma
  \ref{lem:ibgp-sum-sufficient} that there is a valid {\sc iBGP-Sum}
  solution of size at most $\ell \beta + \ell + (m+n+4)^2$.  In the
  second cast we know from Lemma \ref{lem:ibgp-sum-necessary} that any
  valid {\sc iBGP-Sum} solution must have size at least $\ell \beta
  \ln m$.  If we set $\ell$ = $(m+n+4)^2$ this gives a gap of $\ell
  \beta \ln m / \ell(\beta + 2) = \beta \ln m / \beta + 2 =
  \Omega(\log m)$.  The number of vertices $N$ in the {\sc iBGP-Sum}
  instance is $O((m+n+4)^2)$ so $\log m = \Omega(\log N)$, and thus we
  get $\Omega(\log n)$ hardness of approximation.
\end{proof}

It is also fairly simple to modify the basic gadget to prove the same logarithmic hardness for {\sc iBGP-Degree}.  We do this by duplicating everything \emph{other} than $x$, instead of duplicating $x$.  This will force $x$ to have the largest degree.

\begin{theorem}
\label{thm:iBGP_deghard}
  It is NP-hard to approximate {\sc iBGP-Degree} to a factor better than
  $\Omega(\log N)$, where $N$ is the number of vertices in the
  metric.
\end{theorem}
\begin{proof}
  We will use multiple copies of the above gadget.  Let $\alpha$ be some
  large integer that we will define later.  We create $\alpha$ copies
  of the gadget but identify all of the $x$ vertices, so there is
  still a unique $x$ but for all other nodes $v$ in the original there
  are now $\alpha$ copies $v^1, v^2, \dots, v^{\alpha}$.  The distance
  between two nodes in the same copy is exactly as in the original
  gadget, and the distance between two nodes in different copies (say
  $s^i$ and $t^j$) is the distance implied by forcing them to go
  through $x$ (i.e.~$d(s^i, t^j) = d(s,x) + d(x,t)$).  Call this
  metric $M = (V,d)$.  Every vertex in copy $i$ is closer to the rest
  of copy $i$ than to any vertex in copy $j$, so Lemma
  \ref{lem:gadget_necessary} holds for every copy.  Thus if the
  smallest hitting set is $H$ the degree of $x$ in any feasible
  solution to {\sc iBGP-Degree} on $M$ must be at least $\alpha |H|$.

  Conversely, we claim that there is a feasible solution to {\sc
    iBGP-Degree} in which every vertex has degree at most
  $\alpha(|H|+1)$.  Consider the solution in which $x$ is adjacent to
  $z^j$ and to $a_i^j$ for all $j \in [\alpha]$ and $i \in H$, and all
  nodes (other than $x$) in copy $j$ are adjacent to all other nodes
  (other than $x$) in copy $j$ for all $j \in [\alpha]$.  By the above
  analysis of $S(x, b_{T_j}^i)$ we know that this solution satisfies
  these safe sets (via the safe path $x-a_i-b_{T_j}$ where $i \in H$
  is an element in $T_j$).  It also obviously satisfies pairs not
  involving $x$ in the same copy, since there is an edge directly
  between them.  It remains to show that pairs involving $x$ are
  satisfied and that pairs involving two different copies are
  satisfied.

  For the first of these we will show that $z$ is in all safe sets of
  the form $S(x, w^i)$ where $w$ is not a $b$ node.  This is easy to
  verify exhaustively.  It is also true that $z$ is
  in all safe sets of the form $S(w^i,x)$ even when $w$ is a $b$ node,
  since all vertices outside the ball $B(w^i, d(w^i,x))$ are in
  different copies and the shortest path from $z$ to any node in a
  different copy must go through $x$.  Thus the path $x-z-w^i$ in our
  solution satisfies both of these safe sets.  Finally, it is again easy to verify that
  pairs in different copies are also satisfied.

  Now by setting $\alpha$ appropriately we are finished.  Each copy
  has $n+m+4$ nodes, so in the feasible solution we have constructed
  the degree of any node other than $x$ is at most $(n+m+4)^2 + 1$.
  If we set $\alpha$ to some value larger than this, say $(n+m+4)^3$,
  we know that the degree of $x$ has to be at least $(n+m+4)^3 |H|$.
  It is known that it is hard to distinguish between hitting set
  instances with hitting sets of size at most $\beta$ and
  those in which every hitting set has size at least $\beta \ln m$ for
  some value $\beta$.  Suppose that we are in the first case, where
  there is a hitting set of size at most $\beta$.  Then we constructed
  a feasible solution to the {\sc iBGP-Degree} problem with maximum
  degree at most $(n+m+4)^3(\beta+1)$.  In the second case, where
  every hitting set has size at least $\beta \ln m$, we showed that
  the degree of $x$ (and thus the maximum degree) must be at least
  $(n+m+4)^3 \beta \ln n$.  This gives a gap of $\beta \ln m / (\beta
  + 1)$, which is clearly $\Omega(\log m)$.  Since the number of
  vertices in the {\sc iBGP-Degree} instance is polynomial in $m$,
  this implies $\Omega(\log N)$-hardness.
  \end{proof}

\section{Constrained Connectivity} \label{sec:CC}
In this section we consider the hardness of the Constrained Connectivity problems and the integrality gaps of the natural LP relaxations.

\subsection{Hardness}
We now show that the {\sc Constrained Connectivity-Sum} and {\sc
  Constrained Connectivity-Degree} problems are both hard to approximate to better than
$2^{\log^{1-\epsilon} n}$ for any constant $\epsilon > 0$.  We do this via a reduction from {\sc
  Min-Rep}, a problem that is known to be impossible to approximate to
better than $2^{\log^{1-\epsilon} n}$ unless $\text{NP} \subseteq
\text{DTIME}(n^{\polylog(n)})$~\cite{Kortsarz:99}. An instance of {\sc
  Min-Rep} is a bipartite graph $G = (U,V,E)$ in which $U$ is
partitioned into groups $U_1, U_2, \dots, U_m$ and $V$ is partitioned
into group $V_1, V_2, \dots, V_m$.  There is a \emph{super-edge}
between $U_i$ and $V_j$ if there is an edge $\{u,v\} \in E$ such that
$u \in U_i$ and $v \in V_j$.  The goal is to find a minimum set $S$ of
vertices such that for all super-edges $\{U_i, V_j\}$ there is some
edge $\{u,v\} \in E$ with $u \in U_i$ and $v \in V_j$ and $u,v \in S$.  Vertices from a group that are in $S$ are called the \emph{representatives} of the group.
It is easy to prove by a reduction from {\sc Label Cover} that {\sc
  Min-Rep} is hard to approximate to better than $2^{\log^{1-\epsilon}
  n}$, and in particular it is hard to distinguish the case when $2m$
vertices are enough (one from each part in the partition for each side
of the graph) from the case when $2m \times 2^{\log^{1-\epsilon}n}$
vertices are necessary~\cite{Kortsarz:99}.

Given an instance of {\sc Min-Rep}, we want to convert it into an
instance of {\sc Constrained Connectivity-Sum}.  We will create a graph
with five types of vertices: $x_j^i$ for $j \in [m]$ and $i \in [d]$; $U$; $V$; $y_j^i$ for $j \in [m]$ and $i \in [d]$; and $z$.
Here the $x$ nodes represent $d$ copies of the groups of $U$ and the
$y$ nodes represent $d$ copies of the groups of $V$, where $d$ is some
parameter that we will define later.  $z$ is a dummy node that we will
use to connect pairs that are not crucial to the analysis.  Given
this vertex set, there will be four types of edges: $\{x_j^i, u\}$ for all $j \in [m]$ and $i \in [d]$ and $u \in
  U_j$; $\{u,v\}$ for all edges $\{u,v\}$ in the original {\sc Min-Rep}
  instance; $\{v, y_j^i\}$ for all $j \in [m]$ and $i \in [d]$ and $v \in
  V_j$; and $\{w,z\}$ for all vertices $w$.

\begin{figure}
\centerline{\scalebox{.75}{\input{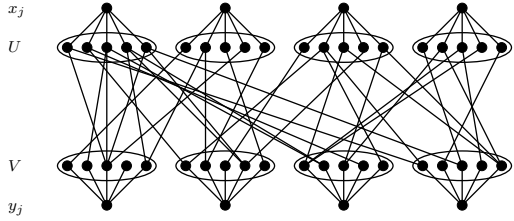}}}
\caption{Basic hardness construction.}
\label{fig:CC_hard}
\end{figure}

This construction is shown in Figure \ref{fig:CC_hard}, except in the actual construction there are $d$ copies of each node in the top and bottom layer and there is a $z$ node that is adjacent to all other nodes.  In Figure \ref{fig:CC_hard} the middle two layers are identical to the original {\sc Min-Rep} problem, and the large ellipses represent the groups.  In the figure we have simply added a new vertex for each group, and in the construction there are $d$ such new vertices per group as well as a $z$ vertex.

Now that we have described the constrained connectivity graph, we need
to define the safe sets.  There are two types of safe sets: if in the
original instance there is a super-edge between $U_i$ and $V_j$ then
$S(x_i^k, y_j^k) = S(y_j^k, x_i^k) = \{x_i^k, y_j^k\} \cup U_i \cup
V_j$ for all $k \in [d]$.  All other safe sets consist of the two
endpoints and $z$.  Let $e_{MR}$ denote the number of super-edges in
the {\sc Min-Rep} instance, let $n_{MR}$ denote the number of
vertices.

The following theorem shows that this reduction works.  The intuition behind it is that a safe path between an $x$ node and a $y$ node corresponds to using the intermediate nodes in the path as the representatives of the groups corresponding to the $x$ and $y$ nodes, so minimizing the number of labels is like minimizing the number of edges incident on $x$ and $y$ nodes.

\begin{theorem} \label{thm:Subgraph_Reduce}
  The original {\sc Min-Rep} instance has a solution of size at most
  $K$ if and only if there is a solution to the reduced Constrained
  Connectivity problem of size at most $Kd + e_{MR} +2md + n_{MR}$.
\end{theorem}
\begin{proof}
  We first prove the only if direction by showing that if there is a
  {\sc Min-Rep} solution of size $K$ then there is a Constrained
  Connectivity solution of size $Kd + e_{MR} + 2md + n_{MR}$.  Let
  $OPT_{MR}$ be the set of vertices in a Min-Rep solution of size $K$.
  Our constrained connectivity solution includes all edges of type
  $4$, i.e.\ we include a star centered at $z$.  For each $i \in [d]$
  and $j \in [m]$ we also include all edges of the form $\{x_j^i, u\}$
  where $u \in U_j \cap OPT_{MR}$ and all edges of the form $\{y_j^i,
  v\}$ where $v \in V_j \cap OPT_{MR}$.  Finally, for each super-edge
  in the Min-Rep instance we include the edge between the pair from
  $OPT_{MR}$ that satisfies it (if there is more than one such pair we
  choose one arbitrarily).  The star clearly has $2md +n_{MR}$ edges,
  there are $Kd$ edges from $x$ and $y$ nodes to nodes in $OPT_{MR}$,
  and there are clearly $e_{MR}$ of the third type of edges, so the
  total number of edges in our solution is $Kd + e_{MR} + 2md +
  n_{MR}$ as required.  To prove that it is a valid solution, we first
  note that for all pairs except those of the form $(x_i^k, y_j^k)$ or
  $(y_j^k, x_i^k)$ where $\{U_i, V_j\}$ is a super-edge are satisfied
  via the star centered at $z$.  For pairs $(x_i^k, y_j^k)$ and
  $(y_j^k, x_i^k)$ with an associated super-edge, since $OPT_{MR}$ is
  a valid solution there must be some $u \in U_i \cap OPT_{MR}$ and $v
  \in V_j \cap OPT_{MR}$ that have an edge between them, and the above
  solution would include that edge as well as the edge from $x_i^k$ to
  $u$ and from $y_j^k$ to $v$, thus forming a safe path of length $3$.

  For the if direction we need to show that if there is a Constrained
  Connectivity solution of size $Kd + e_{MR} + 2md + n_{MR}$ then
  there is a Min-Rep solution of size at most $K$.  Let $OPT_{CC}$ be
  a constrained connectivity solution with $Kd + e_{MR} + 2md +
  n_{MR}$ edges.  Since $S(w,z) = \{w,z\}$ for all vertices $w$, $2md
  + n_{MR}$ of those edges must be a star centered at $z$, so only $Kd
  + e_{MR}$ edges are between other vertices.  Obviously there need to
  be at least $e_{MR}$ edges between $U$ and $V$, since otherwise it
  would be impossible to satisfy all of the demands between $x$ and
  $y$ nodes corresponding to super-edges.  Thus there are at most $Kd$
  edges incident on either $x$ or $y$ nodes.  We can partition these
  edges into $d$ parts, where the edges in the $i$th part are those
  incident on an $x^i$ or $y^i$ node.  So there must be one part of
  size at most $K$; let $i$ be this part.  But since this is a valid
  constrained connectivity solution there is a safe path between
  $x_j^i$ and $y_\ell^i$ for all $j, \ell$ such that there is a
  super-edge between $U_j$ and $y_{\ell}$, and thus the nodes in $U$
  and $V$ that are incident to edges in this $i$th part must form a
  valid Min-Rep solution of size at most $K$.
\end{proof}

We can now set $d = n_{MR}^2$, which gives the following theorem:

\begin{theorem} \label{thm:ccsumhard}
  {\sc Constrained Connectivity-Sum} cannot be approximated better
  than $2^{\log^{1-\epsilon} n}$ for any $\epsilon > 0$ unless $\text{NP} \subseteq
  \text{DTIME}(n^{\polylog(n)})$
\end{theorem}
\begin{proof}
  We know that it is hard to distinguish between an instance of
  {\sc Min-Rep} with a solution of size at most $2m$ and an instance in
  which every solution is of size at least $2m \times
  2^{\log^{1-\epsilon} n}$.  Let $d = n_{MR}^2$.  Then Theorem
  \ref{thm:Subgraph_Reduce} implies that it is hard to distinguish
  between an instance of constrained connectivity with a solution of
  size at most $2mn_{MR}^2 + e_{MR} + 2mn_{MR}^2 + n_{MR} = O(m
  n_{MR}^2)$ and an instance in which every solution has size at least
  $2m2^{\log^{1-\epsilon} n_{MR}} n_{MR}^2 + e_{MR} + 2mn_{MR}^2 +
  n_{MR} = \Omega(m n_{MR}^2 2^{\log^{1-\epsilon} n_{MR}})$.  This
  gives an inapproximability gap of $\Omega(2^{\log^{1-\epsilon}
    n_{MR}})$.  Since $d = n_{MR}^2$ the number of vertices $n$ in our
  constrained connectivity instances is $n_{MR} + 2mn_{MR}^2 \leq
  O(n_{MR}^2)$, and thus $\Omega(2^{\log^{1-\epsilon} n_{MR}}) =
  2^{\Omega(\log^{1-\epsilon} n)}$.  To get this to $2^{\log^{1-\epsilon} n}$ we can simply use a smaller $\epsilon'$.
\end{proof}

We will now prove that {\sc Constrained Connectivity-Degree} has the same hardness of approximation of {\sc Constrained Connectivity-Sum}.  The reduction from {\sc Min-Rep} to the degree problem is basically the same as the reduction to the sum problem, except there are also $d^2$ additional copies of the gadget other than the $x$ and $y$ nodes.  More formally, now the nodes are $x_j^i$ and $y_j^i$ for $j \in [m]$ and $i \in [d]$, $u^{ij}$ for $u \in U$ and $i,j \in [d]$, $v^{ij}$ for $v \in V$ and $i,j \in [d]$, and $z^{ij}$ for $i,j \in [d]$.  Now intuitively each copy $ij$ of the original $U,V$, and $z$ is hooked together exactly like in the original construction, and is hooked up to the nodes $\{x_k^i\}_{k \in [m]}$ and $\{y_k^j\}_{k \in [m]}$ exactly as if they were one copy of the outer $x$ and $y$ nodes of the original construction.

More formally, the edges are the same as before, except now each of the $d^2$ new copies is independent.  In other words, there is an edge between $x_j^i$ and $u^{ik}$ for all $i,k \in [d]$ and $j \in [m]$ and $u \in U_j$, an edge between $y_j^i$ and $v^{ki}$ for all $i,k \in [d]$ and $j \in [m]$ and $v \in V_j$, an edge between $u^{ij}$ and $v^{ij}$ for all $i,j \in [d]$ and edges $\{u,v\}$ in the original {\sc Min-Rep} instance, an edge between $x_j^i$ and $z^{ik}$ for all $i,k \in [d]$ and $j \in [m]$, an edge between $y_j^i$ and $z^{ki}$ for all $i,k \in [d]$ and $j \in [m]$, an edge between $u^{ik}$ and $z^{ik}$ for all $i,k \in [d]$ and $u \in U$, and an edge between $v^{ik}$ and $z^{ik}$ for all $i,k \in [d]$ and $v \in V$.  Similarly, the safe sets are as before but defined by the copies.  That is, $S(x_i^k, y_j^{\ell}) = S(y_j^{\ell}, x_i^k) = \{x_i^k, y_j^{\ell}\} \cup U_i^k \cup V_j^{\ell}$.  All safe sets between nodes in the same copy $ij$ are the two endpoints together with $z^{ij}$, and the safe set of vertices in different copies is just all vertices.

\begin{theorem} \label{thm:ccdeghard}
{\sc Constrained Connectivity-Degree} cannot be approximated better than $2^{\log^{1-\epsilon} n}$ for any constant $\epsilon > 0$ unless $\text{NP} \subseteq \text{DTIME}(n^{\polylog(n)})$
\end{theorem}
\begin{proof}
Every vertex in $U^{ij}$ or $V^{ij}$ can have degree at most $n_{MR} + 1$, since there are only $n_{MR}-1$ other nodes in its copy, and it can in addition be adjacent to $z^{ij}$ and and the node $x^i_k$ or $y^j_k$ corresponding to its group $U_k$ and $V_k$ respectively.  Every node $z^{ij}$ has degree at most $n_{MR} + 2m < 3n_{MR}$, since it can be adjacent to $n_{MR}$ nodes in $U^{ij}$ and $V^{ij}$ as well as $m$ nodes from $X^i$ and $m$ nodes from $Y^j$.  On the other hand, every $x^i_k$ node and every $y^j_k$ node must be adjacent to at least $1$ $U^{i\ell}$ or $V^{\ell j}$ node respectively for all $d$ possibilities for $\ell$.  So every such $x$ or $y$ nodes has degree at least $d$, so if we set $d = 3n_{MR}$ we know that the node with maximum degree must be an $x$ or a $y$ node.

Recall that it is hard to distinguish {\sc Min-Rep} instances with solutions of size at most $2m$ from those in which all solutions have size at least $2m 2^{\log^{1-\epsilon} n_{MR}}$.  Suppose that there is a solution of size $2m$, i.e.\ there is a solution with one representative from each group.  Then there is a solution to the corresponding {\sc Constrained Connectivity-Degree} instance with max degree at most $d$: every $x^i_j$ and $y^i_j$ is connected to its corresponding representative in each of the $d$ copies corresponding to it as well as to the $z$ node for that copy, and in each copy $ij$ we include all edges between $U^{ij}$ and $V^{ij}$ and all edges between those nodes and $z^{ij}$.  It is easy to see that this is a valid solution: by the analysis of Theorem \ref{thm:Subgraph_Reduce} we know that it is valid inside of each copy, and to get between copies nodes $s^{ij}$ and $t^{k\ell}$ can use the safe path $s^{ij} - z^{ij} - x^i_h - z^{i\ell} - y^{k\ell}_h - z^{k\ell} - t^{k\ell}$, where $s$ and $t$ are arbitrary nodes in the copy $ij$, and $h$ is an arbitrary index in $[m]$.

On the other hand, suppose that every solution to the {\sc Min-Rep} instance has size at least $2m2^{\log^{1-\epsilon} n_{MR}}$.  Then as in the analysis of Theorem \ref{thm:Subgraph_Reduce} for every copy $ij$ there must be at least $2m2^{\log^{1-\epsilon} n_{MR}}$ edges that are either between $X^i$ and $U^{ij}$ or between $Y^j$ and $V^{ij}$.  Thus there are at least $d^2 2m2^{\log^{1-\epsilon} n_{MR}}$ such edges.  Since there are only $2md$ vertices in $X \cup Y$, at least one such vertex must have degree at least $2^{\log^{1-\epsilon} n_{MR}} d$.

This shows that it is hard to approximate {\sc Constrained Connectivity-Degree} to better than $2^{\log^{1-\epsilon} n_{MR}} d / d = 2^{\log^{1-\epsilon} n_{MR}}$.  Since the number of vertices $n$ in our instances is polynomial in $n_{MR}$, this means that it is hard to approximate to better than $2^{\Omega(\log^{1-\epsilon} n)}$.  We can then get this to $2^{\log^{1-\epsilon} n}$ by just using a smaller $\epsilon'$.
\end{proof}

\subsection{Integrality Gap for Constrained Connectivity}

We claim that the integrality gap of the flow and cut LP relaxations is large for both {\sc
  Constrained Connectivity-Sum} and {\sc Constrained Connectivity-Degree}.  The intuition is that we use a {\sc Min-Rep} instance in which the edges between each group form a matching (allowing the LP to cheat by breaking up the flow) but many representatives are needed for a valid solution.  This instance is then changed into a Constrained Connectivity problem as in the hardness reduction.  These results are in many ways similar to the $\Omega(n^{1/3 - \epsilon})$ integrality gap for {\sc Min-Rep} recently proved by Charikar et al.~\cite{CHK09}, but the reduction to Constrained Connectivity adds extra complications.

  The instances for
which we will show a large integrality gap are derived from instances
of the \emph{Unique Games problem}, in which we are given a graph
$G=(V,E)$ and a set of permutations $\pi_{uv}$ on some alphabet
$\Sigma$ (one constraint for every edge $(u,v) \in E$) and are asked
to assign a value $x_u$ from $\Sigma$ to each vertex $u$ so as to satisfy
the maximum number of constraints of the form $\pi_{uv}(x_u) = x_v$.
This problem was first considered by Khot~\cite{khot:02}, who
conjectured that it was NP-hard to distinguish instances on which
$1-\delta$ fraction of the constraints can be satisfied from instances
on which at most $\epsilon$ fraction of the constraints can be
satisfied (for sufficiently small $\epsilon$ and $\delta$).  For our
purposes we will consider a minimization version of the Unique Games
problem in which we can assign multiple labels to vertices and the
goal is to assign as few labels as possible so that for every edge
$(u,v)$ there is some label $x_u$ assigned to $u$ with $\pi_{uv}(x_u)$
assigned to $v$.  We first show that there exist instances that
require many labels:

\begin{lemma} \label{lem:Unique} For any constant $\epsilon < 1$,
  there are instances of Unique Games with alphabet size
  $O\left(n^{\frac{2(1+\epsilon)}{1-3\epsilon}}\right)$ and $\Theta(n^2)$ edges that
  require $\epsilon n^2$ labels for any valid solution.
\end{lemma}
\begin{proof}
  We will prove this by the probabilistic method, i.e.\ we will
  analyze a \emph{random} Unique Games instance with the given
  parameters and show that the probability that it has a solution of
  size at most $O(n^2)$ is strictly less than $1$.  This then implies
  the existence of such an instance.  For our random instance, the
  underlying graph will be $K_n$, so there is a permutation constraint
  on every pair of vertices.  Let $k = |\Sigma|$ be the size of the
  alphabet (we will later set this to the value claimed in the lemma, but for now we will leave it
  as a parameter).  For each pair of vertices we will then select a
  permutation uniformly at random from $S_k$.

  Now consider some fixed set $S$ of $\alpha n$ labels (so the average
  number of labels per node is $\alpha$).  What is the probability
  that $S$ is a valid solution?  By Markov's inequality, we know that
  at most $n/2$ vertices have more than $2\alpha$ labels, so there are
  at least $n/2$ vertices with at most $2\alpha$ labels.  Call these
  vertices \emph{light}, and call an edge \emph{light} if both of its
  endpoints are light.  Let $\{u,v\}$ be a light edge.  We claim that
  the probability that $S$ satisfies $\{u,v\}$ is at most
  $\frac{4\alpha^2}{k}$.  To see this, let $\ell \in \Sigma$ be one of
  the labels assigned to $u$ by $S$.  Since the permutation for
  $\{u,v\}$ was chosen uniformly at random, the probability that
  $\ell$ is matched to one of the labels assigned to $v$ by $S$ is at
  most $2\alpha / k$.  Now we can do a union bound over all such
  labels $\ell$, of which there are at most $2\alpha$, to get that the
  probability that edge $\{u,v\}$ is satisfied by $S$ is at most
  $\frac{4\alpha^2}{k}$.  Since the permutations for each edge are
  chosen independently, the event that edge $e$ is satisfied is
  independent of the event that edge $e'$ is satisfied for all $e'
  \neq e$.  Thus the probability that $S$ satisfies \emph{every} edge
  is at most the product of the probabilities that it satisfies each
  fixed edge, i.e.\ the probability that $S$ is a valid solution is at
  most $\left(\frac{4\alpha^2}{k}\right)^{{n \choose 2}} <
  \left(\frac{4\alpha^2}{k}\right)^{\frac{1-\epsilon}{2}n^2}$ (for sufficiently large $n$).

  By the trivial union bound, we know that the probability that there
  is \emph{some} valid solution of size $\alpha n$ for our random
  instance is at most the sum over all possible solutions of size
  $\alpha n$ of the probability that the solution is valid, which by
  the above analysis we know is at most $|\{S : |S| = \alpha n\}|
  \times \left(\frac{4\alpha^2}{k}\right)^{{n \choose 2}}$.  So we
  will now bound $N = |\{S : |S| = \alpha n\}|$, which is easy to do
  by a simple counting argument.  In particular, it is obvious that $N
  = {kn \choose \alpha n}$, since there are exactly $kn$ total labels
  and we are just choosing $\alpha n$ of them.  Now standard bounds
  for binomial coefficients imply that $N \leq \left(\frac{kne}{\alpha
      n}\right)^{\alpha n} = \left(\frac{ke}{\alpha}\right)^{\alpha
    n}$.  Combining this with the previous analysis and setting
  $\alpha = \epsilon n$, we get that the probability that there is
  some valid solution of size $\alpha n$ is at most

  \begin{align*}
    \left(\frac{ke}{\alpha}\right)^{\alpha n} \times \left(\frac{4\alpha^2}{k}\right)^{\frac{1-\epsilon}{2}n^2} &= \frac{4^{\frac{1-\epsilon}{2}n^2} e^{\epsilon n^2} \alpha^{n^2}}{k^{\frac{1-3\epsilon}{2}n^2}}\\
    & = \frac{4^{\frac{1-\epsilon}{2}n^2} e^{\epsilon n^2} \epsilon^{n^2} n^{n^2}}{k^{\frac{1-3\epsilon}{2}n^2}} \\
    & < \frac{n^{(1+\epsilon)n^2}}{k^{\frac{1-3\epsilon}{2}n^2}}
  \end{align*}

  The final inequality is true as long as $n$ is sufficiently large.  If
  we set $k = n^{\frac{2(1+\epsilon)}{1-3\epsilon}}$ then this expression is
  less than $1$.  Since this is the probability that the random Unique
  Games instance we selected has a satisfying solution of size $\alpha
  n$, this implies that for the given parameters there is \emph{some}
  unique games instance that requires more than $\alpha n = \epsilon
  n^2$ labels.
  \end{proof}

  Now that we have found a Unique Games instance that requires many
  labels we would like to use it to construct a {\sc Constrained
  Connectivity-Sum} instance on which the flow LP has large integrality
  gap.  We will basically use the same transformation that we used in
  the reduction of {\sc Min-Rep} to {\sc Constrained Connectivity-Sum}.  Let
  $V_{UG}$ be the vertex set of the above Unique Games instance, and
  let $\Sigma$ be the alphabet.  Then our {\sc Constrained Connectivity-Sum}
  instance will have vertex set $V$ equal to the disjoint union of
  $V_{UG} \times [d]$, $V_{UG} \times \Sigma$, and a special node $z$,
  where $d$ is a duplication parameter that we will set later.  For
  ease of notation, we will let $x_i$ denote the $i$'th copy of vertex
  $x$ in $V_{UG} \times [d]$, i.e.\ $x_i = (x,i)$.  For all $x \in
  V_{UG}$ and $i \in [d]$ there is an edge from $x_i$ to every vertex
  in $x \times \Sigma$.  For every $x,y \in V_{UG}$ and $\alpha, \beta
  \in \Sigma$ there is an edge between $(x, \alpha)$ and $(y, \beta)$
  if and only if assigning $\alpha$ to $x$ and $\beta$ to $y$ is
  sufficient to satisfy the $\{x,y\}$ edge in the Unique Games
  instance (i.e.\ the permutation for that edge matches them up).
  There is also an edge between every vertex and $z$.  For $x,y \in
  V_{UG}$ and $i \in [d]$ we set $S(x_i, y_i) = S(y_i, x_i) = \{x,y\}
  \cup (x \times \Sigma) \cup (y \times \Sigma)$, and we set all other
  safe sets to the two endpoints and $z$.

  \begin{lemma} \label{lem:LP_gap} The value of the flow LP on the
    above {\sc Constrained Connectivity-Sum} instance is at most $2d|V_{UG}| +
    |\Sigma||V_{UG}| + {|V_{UG}| \choose 2}$.
  \end{lemma}
  \begin{proof}
    We prove this by constructing an LP solution of the required size.
    We first set the capacity of every edge incident on $z$ to $1$,
    for a total cost of $|\Sigma||V_{UG}| + d|V_{UG}|$.  This is
    enough capacity to satisfy all pairs other than those of the form
    $(x_i, y_i)$ or $(y_i, x_i)$, since for any other pair $z$ is in
    the safe set so we can send one unit of flow on the edge from one
    endpoint to $z$ and then one unit of flow on the edge from $z$ to
    the other endpoint.

    Now we set the capacity of every other edge to $1/|\Sigma|$.
    Since the number of other edges is $d|V_{UG}||\Sigma| + {|V_{UG}|
      \choose 2} |\Sigma|$ this costs us $d|V_{UG}| + {|V_{UG}|
      \choose 2}$ more, which when added to the cost of the edges to
    $z$ gives us the claimed total LP value.  So we just need to prove
    that this is enough capacity to satisfy demands between $x_i$ and
    $y_i$ for all $x,y \in V$ and $i \in [d]$.  But this is easy to
    see: $x_i$ can send $1/|\Sigma|$ flow to every node in $x \times
    \Sigma$ (for a total flow of $1$), and each of these nodes will
    forward its incoming flow to its neighbor in $y \times \Sigma$.
    Since this is a Unique Games instance this neighbor will be
    unique, and each node in $y \times \Sigma$ will have exactly $1 /
    |\Sigma|$ incoming flow, which it can then forward along its edge
    to $y_i$.  Thus we have enough capacity to send one unit of flow
    from $x_i$ to $y_i$.  And $y_i$ can send flow to $x_i$ the same
    way, just in reverse.
  \end{proof}

  \begin{lemma} \label{lem:IP_gap} Any integral solution to the above
    {\sc Constrained Connectivity-Sum} instance must have size at least $(d
    \times OPT_{UG}) + {|V_{UG}| \choose 2} + d|V_{UG}| +
    |\Sigma||V_{UG}|$ where $OPT_{UG}$ is the minimum number of labels
    needed to satisfy the original Unique Games instance.
  \end{lemma}
  \begin{proof}
    The safe set of any node and $z$ is only that node and $z$, so all
    edges incident to $z$ need to be present in any integral solution
    for a cost of $d|V_{UG}| + |\Sigma||V_{UG}|$.  Furthermore, for
    every pair $u,v \in V_{UG}$ at least one edge must be present from
    $(u \times \Sigma)$ to $(v \times \Sigma)$ since if no such edge
    existed there would be no way of connecting $u_i$ and $v_i$
    through $S(u_i, v_i)$ for any $i \in [d]$.  This adds ${|V_{UG}|
      \choose 2}$ to the total cost, so now we just need to prove that
    there must be at least $d OPT_{UG}$ edges between $(V_{UG} \times
    [d])$ and $(V_{UG} \times \Sigma)$.

    To show this, we will consider some arbitrary integral solution
    and partition the edges between $(V_{UG} \times [d])$ and $(V_{UG}
    \times \Sigma)$ into $d$ parts where the $i$th part consists of
    those edges incident on nodes $\{x_i : x \in V_{UG}\}$.  If every
    part has size at least $OPT_{UG}$ then we are finished.  To prove
    that this is indeed the case, we will prove that for every part,
    the endpoints that are in $V_{UG} \times \Sigma$ actually form a
    valid solution to the Unique Games instance.  So consider the
    $i$th part of the partition.  Suppose that the associated label
    assignment does not form a valid solution to the Unique Games
    instance.  Then there is some pair $u,v \in V$ such that none of
    the labels assigned to $u$ and none of the labels assigned to $v$
    are matched to each other in the permutation corresponding to edge
    $\{u,v\}$.  But this clearly implies that there is no safe path
    from $u_i$ to $v_i$, as any such path must be of length $3$ and
    pass through a label for $u$ and a label for $v$ that are matched
    to each in the permutation corresponding to edge $\{u,v\}$.  This
    is a contradiction since the integral solution must be a valid
    solution.
  \end{proof}

  \begin{theorem} \label{thm:gap} The flow LP for {\sc Constrained Connectivity-Sum} has an integrality gap
    of $\Omega(n^{\frac{1}{3} - \epsilon})$ for any constant $\epsilon
    > 0$.
  \end{theorem}

  \begin{proof}
    We will use the Unique Games instance of Lemma \ref{lem:Unique}
    in the above reduction.  Lemma \ref{lem:LP_gap} implies that the
    flow LP has value at most $O(d |V_{UG}| +
    |V_{UG}|^{\frac{3-\epsilon}{1-3\epsilon}})$ and Lemma
    \ref{lem:IP_gap} implies that any integral solution has size at
    least $\Omega(d\epsilon|V_{UG}|^2) +
    |V_{UG}|^{\frac{3-\epsilon}{1-3\epsilon}})$.  If we let $d = |\Sigma| =
    |V_{UG}|^{\frac{2(1+\epsilon)}{1-3\epsilon}}$ then this gives us an
    integrality gap of
  \begin{equation*}
    \Omega\left(\frac{\epsilon |V_{UG}|^{\frac{4-4\epsilon}{1-3\epsilon}}}{|V_{UG}|^{\frac{3-\epsilon}{1-3\epsilon}}}\right) = \Omega\left(\epsilon|V_{UG}|\right).
  \end{equation*}
  It is easy to see that the number of nodes $n$ in our reduction
  equals $d|V_{UG}| + |\Sigma||V_{UG}|+1$ which in this case is
  $\Theta(|V_{UG}|^{\frac{3-\epsilon}{1-3\epsilon}})$.  Thus the integrality
  gap is $\Omega(n^{\frac{1-3\epsilon}{3-\epsilon}})$, which is
  sufficient since we can set $\epsilon$ to be arbitrarily small.
\end{proof}

We can modify this construction to show a polynomial integrality gap for the flow LP for {\sc Constrained Connectivity-Degree} also.  We will need Unique Games instances with the same parameters as in Lemma \ref{lem:Unique} but on the complete bipartite graph rather than the complete graph.  It is easy to see that Lemma \ref{lem:Unique} can be modified to prove the existence of these instance.  Now the modification is basically the same as the modification we made to show hardness: we just make $d^2$ copies of the inner Unique Games instance and connect them up to the $d$ copies of the outer $x_i$ and $y_i$ nodes in the obvious way.

\begin{theorem} \label{thm:degree_gap}
The flow LP for {\sc Constrained Connectivity-Degree} has an integrality gap of $\Omega(n^{\frac19 - \epsilon})$ for any constant $\epsilon > 0$.
\end{theorem}
\begin{proof}
The maximum degree of any node other than the outer $d$ copies of the $x$ and $y$ nodes is at most $2 |V_{UG}|^{\frac{3-\epsilon}{1-3\epsilon}}$, so if we set $d$ equal to that value we know that the maximum degree must be achieved by some copy of an $x_i$ or $y_i$.  By splitting up the flow equally as in the proof of Lemma \ref{lem:LP_gap} we know that there is an LP solution in which the maximum degree is at most $d|\Sigma| / |\Sigma| + d = 2d$ (where the extra $d$ factor is due to being adjacent to all associated $z$ copies).  On the other hand, we know that any valid integer solution must use at least $\epsilon |V_{UG}|^2$ edges incident on copies of $x_i$ or $y_i$ nodes for each of the $d^2$ instances.  Thus there are at least $d^2 \epsilon |V_{UG}|^2$ edges incident on these nodes in total, and since there are $d|V_{UG}|$ such nodes there must be at least one with degree at least $\epsilon d |V_{UG}|$.  Thus the integrality gap is at least $\epsilon d |V_{UG}| / d = \epsilon |V_{UG}|$.  The total number of nodes in our {\sc Constrained Connectivity-Degree} instance is $O(|V_{UG}| |\Sigma| d^2) = O(|V_{UG}|^{\frac{9-3\epsilon}{1-3\epsilon}})$, so this means the integrality gap is $\Omega(n^{\frac{1-3\epsilon}{9-3\epsilon}})$.  By setting $\epsilon$ small enough this gives us the claimed gap of $\Omega(n^{\frac19-\epsilon})$.
\end{proof}

\section{Hierarchical and Symmetric Safe Sets} \label{sec:hierarchical}

While the constraint the $G = K_n$ gave us some extra power for the iBGP problems, we did not leverage the structure of the safe sets in any way.  In this section we get rid of the requirement on $G$, but show that if the safe sets have an extremely nice structure then {\sc Constrained Connectivity-Sum} can actually be solved optimally in polynomial time.  In the
hierarchical and symmetric safe set version of {\sc Constrained Connectivity-Sum}, $S(x,y) = S(y,x)$ for all $x,y \in V$ and
if some node $z \in S(x,y)$ then $S(x,z) \subseteq S(x,y)$ and $S(z,y)
\subseteq S(x,y)$.  We show that a simple greedy algorithm solves this
version optimally.

We say that a pair $\{x,y\}$ is an \emph{easy} pair if there is some
node $z \in S(x,y)$ such that $S(x,z) \subset S(x,y)$ and $S(y,z)
\subset S(x,y)$.  The pair $\{x,y\}$ is \emph{hard} otherwise.  Note
that in a hard pair $\{x,y\}$, every node $z$ in $S(x,y)$ has either
$S(x,z) = S(x,y)$ or $S(y,x) = S(x,y)$ by the hierarchy property.

\begin{lemma} \label{lem:hierarchy_easy}
  Let $G$ be a graph that has a safe path for all hard pairs.  Then
  all easy pairs also have a safe path in $G$, i.e.\ $G$ is a feasible
  solution.
\end{lemma}
\begin{proof}
  We prove that every pair $\{x,y\}$ has a safe path in $G$ by
  induction on the size of safe sets.  For the base case, all pairs
  $\{x,y\}$ with $|S(x,y)| = 2$ are hard, so by assumption they have a
  safe path in $G$.  For the inductive step, suppose that there are
  safe paths for all pairs $\{u,v\}$ with $|S(u,v)| < k$, and let
  $\{x,y\}$ be a pair with $|S(x,y)| = k$.  If $\{x,y\}$ is hard then
  by assumption there is a safe path.  If it is easy, then there is
  some node $z \in S(x,y)$ such that $S(x,z) \subset S(x,y)$ and
  $S(y,z) \subset S(x,y)$.  Since these two subsets are strictly
  smaller, by induction there is an $x-z$ path contained in $S(x,z)
  \subset S(x,y)$ and there is a $z-y$ path contained in $S(y,z)
  \subset S(x,y)$.  Concatenating these paths give an $x-y$ path
  contained in $S(x,y)$.
\end{proof}

This lemma means that we don't have to worry about satisfying easy
pairs, just hard ones.  We now prove a few structural lemmas that will be useful when designing an algorithm.

\begin{lemma} \label{lem:hierarchy_strong_containment} Let $\{x,y\}$
  be a hard pair.  Then $S(u,v) \subseteq S(x,y)$ for all $u,v \in
  S(x,y)$.
\end{lemma}
\begin{proof}
  Since $\{x,y\}$ is hard either $S(u,x) = S(x,y)$ or $S(u,y) =
  S(x,y)$.  Without loss of generality we assume that $S(u,x) =
  S(x,y)$.  This implies that $v \in S(u,x)$, so by the hierarchy
  property we know that $S(u,v) \subseteq S(u,x) = S(x,y)$.
\end{proof}

This clearly implies that if $G$ is a feasible solution and $\{x,y\}$
is a hard pair then $G|_{S(x,y)}$ is connected and all pairs $u,v \in
S(x,y)$ have a safe path contained in $S(x,y)$.  We now prove some
lemmas about the structure of the optimal solution.

\begin{lemma} \label{lem:hierarchy_OPT_hard}
  Every edge $\{x,y\} \in OPT$ is a hard pair.
\end{lemma}
\begin{proof}
  Suppose $\{x,y\}$ is an edge in $OPT$ that is an easy pair.  Then
  there is some $z \in S(x,y)$ such that $S(x,z) \subset S(x,y)$ and
  $S(y,z) \subset S(x,y)$.  Note that $y \not\in S(x,z)$, since if it
  was then by the hierarchy property we would have that $S(x,y)
  \subseteq S(x,z)$, so $S(x,z) = S(x,y)$ contradicting $\{x,y\}$
  being an easy pair.  Similarly, we know that $x \not\in S(y,z)$.
  Since $OPT$ is feasible there is an $x-z$ path in $S(x,z) \subset
  S(x,y)$ and a $z-y$ path in $S(y,z) \subset S(x,y)$, and by the
  previous observation neither of them use the $\{x,y\}$ edge.  So
  there is an $x-y$ safe path in $OPT$ that does not use the $\{x,y\}$
  edge.  Any hard pair $\{u,v\}$ that use the $\{x,y\}$ edge in a safe
  path can just use the path we found through $z$, since by Lemma
  \ref{lem:hierarchy_strong_containment} $S(x,y) \subseteq S(u,v)$.
  Thus if we remove $\{x,y\}$ all of the hard pairs still have a safe
  path, so by Lemma \ref{lem:hierarchy_easy} so do all of the easy
  pairs.  This contradicts $OPT$ being optimal.
\end{proof}

Order all hard pairs in nondecreasing order, breaking ties
arbitrarily.  We say $\{a,b\} \leq \{c,d\}$ if $\{a,b\}$ comes before
$\{c,d\}$ in this ordering.  We partition the edges of $OPT$ as
follows.  Let $e = \{u,v\}$ be an edge in $OPT$, and let $\{x,y\}$ be
the first hard pair in the ordering such that $u \in S(x,y)$ and $v
\in S(x,y)$, and assign $e$ to part $OPT_{\{x,y\}}$.  By Lemma
\ref{lem:hierarchy_OPT_hard} all edges in $OPT$ are hard pairs so this
is a valid partition.  Let $OPT_{\leq \{x,y\}} = \cup_{\{a,b\} \leq
  \{x,y\}} OPT_{\{a,b\}}$, and let $OPT_{< \{x,y\}}$ be defined
analogously.

\begin{lemma} \label{lem:hierarchy_OPT_connected}
  Let $\{x,y\}$ be a hard pair.  Then $OPT_{\leq \{x,y\}} |_{S(x,y)}$
  is connected.
\end{lemma}
\begin{proof}
  Let $\{u,v\}$ be an edge in $OPT|_{S(x,y)}$.  Then since $\{u,v\}$
  is a hard pair (by Lemma \ref{lem:hierarchy_OPT_hard}) and $\{x,y\}$
  is a hard pair with both $u$ and $v$ in $S(x,y)$, by the definition
  of the partition the part $OPT_{\{a,b\}}$ containing $\{u,v\}$ must
  have $\{a,b\} \leq \{x,y\}$.  Thus $\{u,v\} \in OPT_{\leq \{x,y\}}
  |_{S(x,y)}$.
\end{proof}

We now finally give our algorithm.  First we construct the above
ordering.  We then consider hard pairs in this order, and when
considering a pair $\{x,y\}$ we add the minimum number of edges
required to make our current graph restricted to $S(x,y)$ connected.
This algorithm clearly returns a feasible solution, since for any hard
pair $\{x,y\}$ at some point we consider it and make sure that its
safe set is connected and that is sufficient by Lemma
\ref{lem:hierarchy_easy}.  For every hard pair $\{x,y\}$, let
$ALG_{\{x,y\}}$ by the edges added by the algorithm when considering
$\{x,y\}$, and define $ALG_{<\{x,y\}} = \cup_{\{a,b\} < \{x,y\}}
ALG_{\{a,b\}}$ and $ALG_{\leq \{x,y\}}$ analogously.  Now we will
prove that $|ALG| \leq |OPT|$.

\begin{lemma} \label{lem:hierarchy_OPT_ALG}
  The endpoints of any edge in $OPT_{< \{x,y\}}|_{S(x,y)}$ are
  connected in $ALG_{< \{x,y\}}|_{S(x,y)}$.
\end{lemma}
\begin{proof}
  Let $\{u,v\}$ be an edge in $OPT_{< \{x,y\}}|_{S(x,y)}$.  Then
  $\{u,v\} \in OPT_{\{a,b\}}$ for some $\{a,b\} < \{x,y\}$.  By
  definition, this means that $\{a,b\}$ is the first pair in the
  ordering with a safe set that contains both $u$ and $v$.  By Lemma
  \ref{lem:hierarchy_strong_containment} we know that $S(u,v)
  \subseteq S(a,b)$.  We also know that $\{u,v\}$ is a hard pair by
  Lemma \ref{lem:hierarchy_OPT_hard}, so if $S(u,v) \subset S(a,b)$
  then $\{u,v\}$ would be before $\{a,b\}$ in the ordering and would
  contain both $u$ and $v$, contradicting the definition of $\{a,b\}$.
  Thus $S(u,v) = S(a,b)$.  After considering $\{a,b\}$ the algorithm
  guarantees that $ALG_{\leq \{a,b\}}|_{S(a,b)}$ is connected, and
  therefore there is a safe $u-v$ path in $ALG$ after considering
  $\{a,b\}$.  We also know from Lemma
  \ref{lem:hierarchy_strong_containment} that $S(u,v) \subseteq
  S(x,y)$, so this safe path is entirely present in $ALG_{<
    \{x,y\}}|_{S(x,y)}$ and thus $u$ and $v$ are connected in $ALG_{<
    \{x,y\}}|_{S(x,y)}$.
\end{proof}

\begin{theorem} \label{lem:hierarchy_main}
$|ALG| \leq |OPT|$
\end{theorem}
\begin{proof}
  We will prove that $|ALG_{\{x,y\}}| \leq |OPT_{\{x,y\}}|$ for all
  hard pairs $\{x,y\}$.  Since these form a partition of the edges of
  $ALG$ and of $OPT$, this is sufficient to prove that $|ALG| \leq
  |OPT|$.  Consider some such hard pair $\{x,y\}$.  We know from Lemma
  \ref{lem:hierarchy_OPT_connected} that $OPT_{\leq
    \{x,y\}}|_{S(x,y)}$ is connected, so $OPT_{\{x,y\}}$ must contain
  enough edges to connect the components of $OPT_{<
    \{x,y\}}|_{S(x,y)}$.  By the definition of the algorithm,
  $ALG_{\{x,y\}}$ has the minimum number of edges necessary to connect
  the components of $ALG_{< \{x,y\}}|_{S(x,y)}$.  Now since the number
  of components in $ALG_{< \{x,y\}}|_{S(x,y)}$ is at most the number
  of components of $OPT_{< \{x,y\}}|_{S(x,y)}$ (by Lemma
  \ref{lem:hierarchy_OPT_ALG}), this implies that $|ALG_{\{x,y\}}|
  \leq |OPT_{\{x,y\}}|$.
\end{proof}

\newpage

\bibliographystyle{abbrv}
\bibliography{cc}

\end{document}